\documentclass[11pt]{article}

\usepackage{mathpazo}
\usepackage{amsfonts}
\usepackage{amsmath}
\usepackage{graphicx}
\usepackage{latexsym}

\usepackage[margin=1.05in]{geometry}
  
\usepackage[T1]{fontenc}
\usepackage{times}
\usepackage{color,graphicx}
\usepackage{array}
\usepackage{enumerate}
\usepackage{amsmath}
\usepackage{amssymb}
\usepackage{amsthm}
\usepackage{pgfplots}
\usepackage{pgf}
\usepackage{tikz}
\usetikzlibrary{patterns}
\usepackage{filecontents}
\usepgfplotslibrary{patchplots} 
\usetikzlibrary{pgfplots.patchplots} 
\pgfplotsset{width=9cm,compat=1.5.1}

\usepackage{amsthm}

\newtheorem{definition}{Definition}
\newtheorem{lemma}{Lemma}
\newtheorem{theorem}{Theorem}
\newtheorem{corollary}{Corollary}

\newtheorem{proposition}{Proposition}

\newtheorem{example}{Example} 

\newcommand{\CP}{\mathcal{CP}}
\newcommand{\CPCP}{\mathcal{CPCP}}
\newcommand{\CPDNN}{\mathcal{CPDNN}}
\newcommand{\PSD}{\mathcal{PSD}}
\newcommand{\DNN}{\mathcal{DNN}}
\newcommand{\DD}{\mathcal{DD}}

\usepackage{xcolor}
\usepackage{makeidx}
\usepackage[colorlinks=true,linkcolor=blue,anchorcolor=blue,citecolor=red,urlcolor=magenta]{hyperref}

\newcommand{\tr}{\operatorname{Tr}}

\newcommand{\bra}[1]{\langle #1 |}
\newcommand{\ket}[1]{| #1 \rangle}

\newcommand{\ketbra}[2]{| #1 \rangle\langle #2 |}

\newcommand\Tr{\mathop{\rm Tr}\nolimits}

\newcommand{\conv}{\mathrm{conv}}

\newcommand{\defeq}{\stackrel{\smash{\textnormal{\tiny def}}}{=}}

\def\I{\mathbb{1}} 
\def\C{\mathbb{C}}
\def\N{\mathbb{N}}

\def\R{\mathbb{R}}

\def\D{\mathcal{D}}
\def\Lin{\mathcal{L}}

\def\v{\mathbf{v}}
\def\w{\mathbf{w}}

\def\0{\mathbf{0}}

\begin{document}

\title{Completely positive completely positive maps\\(and a resource theory for non-negativity of quantum amplitudes)}  

\author{
	Nathaniel Johnston\footnote{Department of Mathematics \& Computer Science, Mount Allison University, Sackville, NB, Canada E4L 1E4}\textsuperscript{$\ \ $}\footnote{Department of Mathematics \& Statistics, University of Guelph, Guelph, ON, Canada N1G 2W1} \quad and\quad
	Jamie Sikora\footnote{Virginia Polytechnic Institute and State University, Blacksburg, VA, USA 24061}
}

\date{August 16, 2022}

\maketitle

\begin{abstract}
	In this work we examine quantum states which have non-negative amplitudes (in a fixed basis) and the channels which preserve them. 
	These states include the ground states of stoquastic Hamiltonians and they are of interest since they avoid the Sign Problem and can thus be efficiently simulated. 
	In optimization theory, the convex cone generated by such states is called the set of completely positive ($\CP$) matrices (not be confused with completely positive superoperators).
	We introduce quantum channels which preserve these states and call them \emph{completely positive completely positive}. To study these states and channels, we use the framework of resource theories and investigate how to measure and quantify this resource.\\
	
	
\end{abstract}
  

\section{Introduction} 

There are several notions of what it means for a Hermitian matrix $X$ to be ``non-negative''. 
For instance, we could require that $X$ has non-negative eigenvalues; such matrices are called \emph{positive semidefinite (PSD)} and we denote the set of them by $\PSD$.  
If we further wish the entries to be non-negative (in some fixed basis), then such matrices are called \emph{doubly non-negative (DNN)}, the set of which is denoted $\DNN$. 
There are, however, many other useful definitions of non-negativity. 
For example, the set of \emph{completely positive (CP)} matrices, denoted $\CP$, is defined as 
\begin{equation} 
    \CP \defeq \conv \{ \mathbf{x}\mathbf{x}^{\top} : \mathbf{x} \geq 0 \}  
\end{equation}  
where, again, entrywise non-negativity of the vector $\mathbf{x}$ is with respect to a fixed basis, as is the transpose $\mathbf{x}^\top$. It is straightforward to see that we have the inclusions 
\begin{equation} 
    \CP \subseteq \DNN \subseteq \PSD. 
\end{equation}
Moreover, it can be shown that $\CP$ and $\DNN$ are the same if and only if the matrices are of size $4 \times 4$ or smaller~\cite{GW80}.  

Optimizing over the set of doubly non-negative matrices can be done efficiently using semidefinite programming, but, on the other hand, optimizing over $\CP$ is NP-hard~\cite{DG14}. Many natural NP-hard problems can be modelled via an optimization over the set of CP matrices as a cone program, and its relaxation to the set of DNN matrices gives a semidefinite programming relaxation to such problems which, as mentioned above, can be solved efficiently. 

Although CP matrices are well-studied in the optimization community \cite{Ber88,BS03}, they have seen only a few applications in the quantum information theory literature. Some examples of such applications include the fact that a mixed Dicke state is separable if and only if a certain associated matrix is CP \cite{Yu16,TAQLS17}, the related fact that copositive matrices (which are dual to CP matrices) can be used to construct symmetric entanglement witnesses \cite{MATS21}, and the use of completely positive matrices in determining separability or entanglement of quantum states with diagonal unitary symmetries \cite{JM19,SN21}. CP matrices also arise in the study of classical correlations with respect to a non-local game~\cite{PSVW18}.

Indeed, when one utters the words ``completely positive'', a quantum information theorist is almost surely going to think of superoperators. 
A superoperator $\Phi$ is said to be \emph{positive} if it maps PSD matrices to PSD matrices (not necessarily of the same dimension)
and is said to be \emph{completely positive}\footnote{We use $\CP$ to denote the set of completely positive matrices and avoid using any notation for completely positive superoperators to avoid confusion.} if $(\I_k \otimes \Phi)$ is positive for all $k \in \N$, where $\I_k$ is the identity map on $k \times k$ matrices. 
This notion of complete positivity is one of two ingredients in the definition of a quantum channel (the other being  trace-preservation).  
 
In this paper, we combine the two notions of complete positivity and look at superoperators which map CP matrices to CP matrices (not necessarily of the same dimension) as in the following (informal) definition. (A formal definition can be found in Section~\ref{sec:prelims}). 

\begin{definition} 
A superoperator $\Phi$ is said to be a \emph{completely positive completely positive (CPCP)} map if, for all $k \in \N$, we have 
\begin{equation} 
(\I_k \otimes \Phi)(X) \text{ is CP whenever } X \text{ is CP}.  
\end{equation} 
We denote the set of completely positive completely positive maps by $\CPCP$.
\end{definition}  

We show that CPCP maps are also completely positive, and are thus valid quantum channels as long as they are also trace-preserving. 
This paper studies CPCP quantum channels and how they act on particular quantum states, most notably the states that are represented by a CP density matrix, which we now discuss. 

A \emph{density matrix} is a positive semidefinite matrix with unit trace. Density matrices are the most general description of a quantum state, and we identity quantum states with their density matrices. 
If a quantum state $\rho$ has rank $1$ then it can be written as $\mathbf{x}\mathbf{x}^*$ for some column vector $\mathbf{x}$ satisfying $\| \mathbf{x} \|_2 = 1$. 
In this case, we identify the quantum state with $\mathbf{x}$ itself and it is called a \emph{pure} state. 
In this paper we use Dirac notation for pure states, i.e., the notation $\ket{v}$ refers to a column vector/pure quantum state with $\| \ket{v} \|_2 = 1$, and we only use this notation for vectors with unit norm.  
The notation $\bra{v}$ is defined as $\bra{v} \defeq \ket{v}^*$, the conjugate transpose of $\ket{v}$. 
Notice that $\ketbra{v}{v}$ is a rank-$1$ density matrix. 
 
When we say that a quantum state $\rho$ is CP, we mean 
that $\rho \in \CP$ and $\tr(\rho)=1$. 
Indeed, every CP quantum state can be decomposed as 
\begin{equation} 
    \sum_{i=1}^n \lambda_i \, \ketbra{v_i}{v_i},
\end{equation} 
where $\ket{v_1}, \ldots, \ket{v_n} \geq 0$ (with this inequality being meant entrywise), and $\lambda_1, \ldots, \lambda_n \geq 0$ satisfy $\sum_{i=1}^n \lambda_i = 1$. 
We see that $\ket{v} \geq 0$ means that $\ket{v}$ can be written as a non-negative linear combination of the fixed basis vectors, or in quantum terms, has non-negative amplitudes. 
Such quantum states are obviously preserved under CPCP quantum channels. 

One reason that CP quantum states are interesting is because they are ground states of \emph{stoquastic Hamiltonians}. 
A Hamiltonian is represented by a Hermitian matrix $H$ and a quantum state $\rho$ is in its ground space if it satisfies  
\begin{equation} 
\tr(H\rho) = \lambda_{\min}(H), 
\end{equation} 
where $\lambda_{\min}$ denotes the minimum eigenvalue. 
The Hamiltonian $H$ is said to be \emph{stoquastic} if all of its off-diagonal matrix elements (in a fixed basis)  
are real and non-positive. It follows from the Perron-Frobenius theorem that there exists a ground state of such a Hamiltonian that is CP.   
These quantum states are well-suited to Quantum Monte Carlo methods since they avoid the so-called \emph{Sign Problem} and can thus be efficiently simulated~\cite{Ohz17}. 
Dealing with the Sign Problem is one of the biggest challenges in the study of many-body quantum systems. 
Therefore, deciding whether a quantum state is CP (in some basis) is very closely related to the task of deciding whether it can be simulated efficiently using Quantum Monte Carlo methods.  

To give a framework for studying CP quantum states and CPCP quantum channels, we examine them from the perspective of a resource theory.  
Resource theories \cite{CG18} have been very effective for studying other interesting notions in quantum mechanics such as entanglement~\cite{Vid00}, coherence~\cite{BCP14}, and quantum computation with stabilizer states and operations \cite{VMGE14}, to name a few. 
Briefly, a resource theory has a set of \emph{free resources} (in our context, CP quantum states) and a set of \emph{free  operations} which preserve the free states (in our context, CPCP quantum channels).  
A quantum state which is not free is said to be \emph{resourceful}. 
In this paper, we introduce and explore the resource theory for non-negativity of quantum amplitudes that arises from considering these sets of free states and free operations. In particular, we investigate questions such as: 
\begin{itemize} 
    \item How resourceful is a particular state?  
    \item What is the most resourceful state? 
    \item Can we experimentally witness non-freeness of a state? 
    \item Are free operations physically or operationally motivated? 
\end{itemize}

The resource theory we develop in this paper has close connections to other possible and existing resource theories.  
Another definition of a pure CP quantum state $\ket{v}$ is that each of its entries has the same relative phase, i.e., there exists $\theta \in [0, 2\pi)$ such that $e^{i \theta} \ket{v} \geq 0$. 
This is because this phase does not appear in the outer product $\ketbra{v}{v}$. 
So, the resource theory of non-negative amplitudes could equivalently be considered as a resource theory of ``phaseyness'', which is closely related to the already-investigated resource theory of ``imaginarity'' \cite{HG18,WKR21}. Also, since CP quantum states are efficiently simulatible, this resource theory is closely related the resource theory of efficient simulatible quantum states.  
Towards the end of the paper, we compare the resource theory of non-negative amplitudes to the resource theory of coherence. 

\paragraph{Paper organization.} 
We introduce our notation and mathematical preliminaries needed for our results in Section~\ref{sec:prelims}. 
We then start building up our resource theory by exploring its set of free states in Section~\ref{sec:free_states}, which as mentioned earlier are the 
CP quantum states. 
In Section~\ref{sec:copositive} we explore the cone dual to the set of free states, which act for witnesses of this resource. 
In Section~\ref{sec:free_operations} we introduce various families of quantum channels that preserve these free states and thus act as free operations in this resource theory. 
We then explore various ways of measuring how resourceful a state is in Section~\ref{sec:measures_nonneg} and discuss the most resourceful state. 
We then close in Section~\ref{sec:coherence} with some remarks about how this resource theory relates to the resource theory of coherence, and in Section~\ref{sec:conclusions} with some open questions. 


\section{Mathematical preliminaries and notation}\label{sec:prelims}

Our notation and terminology is mostly standard in quantum information theory, so we direct the reader to the books \cite{Wat18,NC00} for a more thorough introduction. We use ``kets'' like $\ket{v},\ket{w} \in \C^n$ to denote complex column vectors with (Euclidean) norm~$1$, which represent pure quantum states, and we use boldface lowercase letters like $\v,\w \in \C^n$ to denote column vectors whose norms perhaps differ from~$1$. The standard (computational) basis of $\C^n$ is $\{\ket{j}\}_{j=0}^{n-1}$. We use $M_{m,n}$ to denote the set of $m \times n$ complex matrices, $M_n = M_{n,n}$, and $\D_n \subset M_n$ to denote the set of $n \times n$ density matrices (i.e., positive semidefinite matrices with trace~$1$, which represent mixed quantum states).

Standard inequality signs are meant entrywise, so that $\ket{v} \geq 0$ means that every entry of $\ket{v} \in \C^n$ is real and non-negative, and $A \geq 0$ means the same for the entries of $A \in M_n$. On the other hand, we use alternative inequality signs, as in $A \succeq 0$ and $A \preceq B$, to refer to the Loewner partial order, i.e., they mean that $A$ and $B-A$ are Hermitian positive semidefinite, respectively.

The set of linear maps from $M_n$ to $M_m$ is denoted by $\Lin(M_n,M_m)$. Of particular interest are the linear maps $\Phi \in \Lin(M_n,M_m)$ that are both \emph{completely positive}, i.e., 
\begin{equation} 
(\I_k \otimes \Phi)(X) \succeq 0 \; \text{ whenever } \; 0 \preceq X \in M_k \otimes M_n \; \text{ and } \; k \in \mathbb{N}, 
\end{equation} 
where $\I_k \in \Lin(M_k,M_k)$ is the identity map, and \emph{trace-preserving}, i.e., 
\begin{equation} 
\tr(\Phi(X)) = \tr(X) \; \text{ for all } \; X \in M_n.  
\end{equation} 
Linear maps with these two properties represent \emph{quantum channels}, and some particularly important examples include the \emph{partial traces} 
\begin{align} 
\mathrm{Tr}_1 : M_m \otimes M_n \rightarrow M_n 
\; \text{ defined by } \; 
\mathrm{Tr}_1(A \otimes B) = \tr(A) B 
\; \text{ for all } \; 
A \in M_m, B \in M_n; \\ 
\mathrm{Tr}_2 : M_m \otimes M_n \rightarrow M_m 
\; \text{ defined by } \; 
\mathrm{Tr}_2(A \otimes B) = \tr(B) A 
\; \text{ for all } \; 
A \in M_m, B \in M_n. 
\end{align}

A linear map $\Phi \in \Lin(M_n,M_m)$ is completely positive if and only if its \emph{Choi matrix}
\begin{equation}
J(\Phi) \defeq \sum_{i,j=0}^{n-1} \ketbra{i}{j} \otimes \Phi(\ketbra{i}{j}) 
\end{equation} 
is positive semidefinite \cite{Cho75}, and trace-preservation of $\Phi$ is equivalent to $\tr_2(J(\Phi)) = I_n$, where $I_n \in M_n$ is the $n \times n$ identity matrix. Equivalently, $\Phi$ is completely positive if and only if it can be written in the form $\Phi(X) = \sum_{j=1}^k A_j X A_j^*$ for some family of matrices $\{ A_1, \ldots, A_k \} \subset M_{m,n}$ called \emph{Kraus operators}, and trace-preservation of $\Phi$ is equivalent to $\sum_{j=1}^k A_j^*A_j = I_n$.  


\section{Free states: Normalized completely positive matrices}\label{sec:free_states}

Every quantum resource theory consists of two components: a subset $\mathcal{F}_n \subseteq \D_n$ consisting of \emph{free states}, which are the states that are ``useless'' at the particular task(s) considered by the resource theory, and a set of quantum channels that send $\mathcal{F}_n$ to $\mathcal{F}_m$ (called \emph{free operations}) \cite{CG18,Reg18}.

For the present resource theory, the free pure states are simply those that have non-negative amplitudes with respect to some fixed basis of $\C^n$ (which we assume is simply the computational basis for convenience). 
Equivalently, since pure states are only defined up to global phase, these are the pure states in which all entries have the same phase as each other. 
In other words, $\ket{v}$ is free if and only if there exists $\theta \in [0, 2 \pi)$ such that $e^{i\theta} \ket{v}$ has real and non-negative amplitudes. 
For mixed states, the free states are simply the convex combinations of (projections onto) these pure states: 
\begin{align}\label{eq:CP_defn}
    \CP_n \defeq \left\{ \sum_{j=1}^m p_j  
    \ketbra{v_j}{v_j} \in \D_n : 
    \ket{v_1}, \ldots, \ket{v_m} \geq 0, \, 
    p_1, \ldots, p_m \geq 0, \, 
    \sum_{j=1}^m p_j = 1, \ m \in \N \right\}.
\end{align}
Equivalently, the set $\CP_n$ consists of the reduced states that are obtainable after tracing out one half of a pure state with non-negative amplitudes:
\begin{align}\label{eq:CP_auxilliary}
    \CP_n = \big\{ \tr_1(\ketbra{v}{v}) \in \D_n : 0 \leq \ket{v} \in \C^m \otimes \C^n, m \in \N \big\}.
\end{align}
Indeed, this equivalence follows immediately from that fact that we can write every $0 \leq \ket{v} \in \C^m \otimes \C^n$ as $\ket{v} = \sum_{j=0}^{m-1} \ket{j} \otimes \sqrt{p_j}\ket{v_j}$ for some non-negative $\{p_j\}$ and entrywise non-negative $\{\ket{v_j}\}$. A pure state $\ket{v} \in \C^m \otimes \C^n$ for which $\Tr_1(\ketbra{v}{v}) = \rho$ is called a \emph{purification} of $\rho$, so Equation~\eqref{eq:CP_auxilliary} says that $\CP_n$ consists exactly of the density matrices that have an entrywise non-negative purification.

Determining whether or not a particular mixed state $\rho$ is free (i.e., determining if $\rho \in \CP_n$) is NP-hard~\cite{DG14}, so we sometimes work with the set of \emph{doubly non-negative} density matrices instead:
\begin{align}\label{eq:DNN_defn}
    \DNN_n \defeq \big\{ \rho \in \D_n : \rho \geq 0 \big\}.
\end{align}
Membership in $\DNN_n$ can be determined straightforwardly, which makes it much easier to work with in many settings. It is straightforward to see that $\CP_n \subseteq \DNN_n$, and it is true (but not straightforward to see) that equality holds if and only if $n \leq 4$ \cite{GW80}. When $n \geq 5$, there are density matrices with non-negative entries that nonetheless are not completely positive (and thus not free in this resource theory), with one simple example \cite{HN63} being
\begin{align}\label{eq:dnn_not_cp}
    \rho = \frac{1}{9}\begin{bmatrix}
        1 & 1 & 0 & 0 & 1 \\
        1 & 2 & 1 & 0 & 0 \\
        0 & 1 & 2 & 1 & 0 \\
        0 & 0 & 1 & 1 & 1 \\
        1 & 0 & 0 & 1 & 3
    \end{bmatrix} \in \DNN_5 \backslash \CP_5.
\end{align}
In particular, the fact that this mixed state is not completely positive means that, despite its entries all being real and non-negative, it does not have a purification with real and non-negative entries. We verify that the state~\eqref{eq:dnn_not_cp} is not completely positive in Section~\ref{sec:copositive} (though this example is well-known in the literature of completely positive matrices).


\subsection{The completely positive rank}\label{sec:cp_rank}

It is worth emphasizing that the number of terms required in the convex sum in Equation~\eqref{eq:CP_defn} (or equivalently, the dimension $m$ needed in Equation~\eqref{eq:CP_auxilliary}) is indeed finite, and we can choose $m \leq n^2+1$ by Carath\'{e}odory's Theorem. 
The minimum number of terms $m$ needed to represent a particular state $\rho \in \CP_n$ in this way is called its \emph{completely positive rank} (CP-rank), and the argument we just provided shows that the CP-rank of every CP state is no greater than $n^2+1$. 
In fact, this bound can be reduced by roughly a factor of $2$: it was shown in \cite{BB03} that the CP-rank of every CP state is at most $n(n+1)/2-1$, but the best possible upper bound is not known in general (for example, when $n = 5$, the maximum CP-rank is $6$ \cite{SBJS13}, which is smaller than the general upper bound of $5(5+1)/2-1 = 14$).

The central problem investigated in \cite{TCF19} was whether or not a given pure quantum state in a tensor product space could be transformed via local unitary operations into one that has non-negative entries. The Schmidt decomposition tells us that for every pure state $\ket{v} \in \C^m \otimes \C^n$, there exist unitary matrices $U_1 \in M_m$ and $U_2 \in M_n$ such that $(U_1 \otimes U_2)\ket{v} = \sum_{j}\lambda_j \ket{j} \otimes \ket{j} \geq 0$, where $\{\lambda_j\}$ are the (non-negative) Schmidt coefficients of $\ket{v}$. The following proposition uses the CP-rank to answer the variant of this question where we only have unitary freedom on one half of the state, rather than both halves.
  
\begin{proposition}\label{prop:cp_rank_local_unitary}
    Suppose $\ket{v} \in \C^m \otimes \C^n$ is a pure state. There exists a unitary matrix $U \in M_m$ such that $(U \otimes I_n)\ket{v} \geq 0$ if and only if $\tr_1(\ketbra{v}{v})$ is completely positive with CP-rank $\leq m$.
\end{proposition}

\begin{proof}
    For the ``only if'' direction, suppose that there is a unitary matrix $U$ such that ${\ket{w} := (U \otimes I_n)\ket{v} \geq 0}$. If we write $\ket{w} = \sum_{j=0}^{m-1}\ket{j} \otimes \w_{\mathbf{j}}$, then $\w_{\mathbf{j}} \geq 0$ for each $j$, then
    \begin{equation}
        \tr_1(\ketbra{v}{v}) = \tr_1(\ketbra{w}{w}) = \sum_{j=0}^{m-1} \w_{\mathbf{j}}\w_{\mathbf{j}}^*,
    \end{equation} 
    which is completely positive with CP-rank $\leq m$.
    
    The ``if'' direction follows by reversing the above argument and noting that given two pure states ${\ket{v},\ket{w} \in \C^m \otimes \C^n}$, they satisfy $\tr_1(\ketbra{v}{v}) = \tr_1(\ketbra{w}{w})$ if and only if there exists a unitary matrix $U \in M_m$ such that $\ket{w} = (U \otimes I_n)\ket{v}$ \cite[Theorem~2.12]{Wat18}.
\end{proof}

However, determining whether or not such a local unitary transformation exists in the multipartite case (i.e., where we have three or more tensor factors, rather than just two) seems much more difficult.

        
\section{Witnesses: Copositive matrices}\label{sec:copositive}

A real symmetric matrix $W \in M_n$ is called \emph{copositive} if $\bra{v}W\ket{v} \geq 0$ whenever $\ket{v} \geq 0$, or equivalently, if $\tr(W \rho) \geq 0$ whenever $\rho \in \CP_n$. 
In other words, copositive matrices are the members of the dual cone of $\CP_n$, if we think of $\CP_n$ as a subset of the vector space of real symmetric $n \times n$ matrices. However, in our setting of quantum information theory, it is much more natural to regard $\CP_n$ as a subset of the \emph{complex Hermitian} $n \times n$ matrices. When we do this, the dual cone $\CP_n^\circ$ actually consists not only of copositive matrices, but also any Hermitian matrix $W$ whose entrywise real part\footnote{We mean ``real part'' entrywise, not in the sense of the Hermitian $+$ skew-Hermitian Cartesian decomposition.} $\Re(W)$ is copositive:
\begin{equation} 
    \CP_n^\circ = \big\{ W = W^* \in M_n : \Re(W) \ \text{is copositive} \big\}.
\end{equation} 
In other words, $\CP_n^\circ$ consists of the copositive matrices plus arbitrary imaginary part (subject to Hermiticity). Indeed, if $W$ is such a matrix then $\tr(W\rho) = \tr(\Re(W)\rho) \geq 0$ whenever $\rho \in \CP_n$ thanks to Hermiticity of $W$ and realness of $\rho$. Note that unlike the members of $\CP_n$, we do not place any normalization condition on the members of $\CP_n^\circ$, as it is a cone.
  
A standard separating hyperplane argument shows that for every density matrix $\rho \notin \CP_n$, there exists $W \in \CP_n^\circ$ for which $\tr(W\rho) < 0$. In fact, $W$ can always be chosen to be real, if desired, and hence a copositive matrix.  
For this reason, we think of copositive matrices as witnesses for the resourcefulness of the non-free state $\rho$ (just like every entangled state can be verified to be entangled via some entanglement witness). 
Furthermore, this method provides a way of demonstrating the resourcefulness of a state that is directly measurable in a lab---$W$ can be thought of as an observable that we measure in a system with state $\rho$, and Born's rule tells us that $\tr(W\rho) < 0$ is the average result of that measurement.

It is clear that if a symmetric matrix $W \in M_n$ can be written as the sum of a positive semidefinite matrix and an entrywise non-negative matrix then it is copositive, since $\tr(W\rho) \geq 0$ for all $\rho \in \CP_n$. 
In fact, the matrices of this form are exactly the real parts of the members of the dual cone of the doubly non-negative density matrices:\footnote{Again, we are considering $\DNN_n$ as a subset of the set of complex Hermitian matrices. If we considered it as a subset of real symmetric matrices, its dual cone $\DNN_n^\circ$ would be the same, except its members would all be real.}   
\begin{equation}
    \DNN_n^\circ = \big\{ W = W^* \in M_n : \Re(W) = X + Y, X \succeq 0, Y \geq 0 \big\}.
\end{equation} 
Since $\CP_n \subseteq \DNN_n$ with equality if and only if $n \leq 4$, it follows immediately from standard results about dual cones (see \cite{BV04}, for example) that $\DNN_n^\circ \subseteq \CP_n^\circ$ with equality if and only if $n \leq 4$ as well.

The most well-known matrix that is copositive but not a member of $\DNN_n^\circ$ is the \emph{Horn matrix} \cite{HN63} $W_1$, where for $x \geq 1$ we define
\begin{align}\label{eq:horn_like_Wx}
    W_x = \begin{bmatrix}
        1 & -1 & x & x & -1 \\
       -1 & 1 & -1 & x & x \\
        x & -1 & 1 & -1 & x \\
        x & x & -1 & 1 & -1 \\
       -1 & x & x & -1 & 1
   \end{bmatrix}.
\end{align}
In fact, the Horn matrix $W_1$ verifies that the density matrix $\rho$ from Equation~\eqref{eq:dnn_not_cp} is indeed not completely positive as we claimed earlier, since $\tr(W_1\rho) = -1/9 < 0$.

To confirm that $W_x \notin \DNN_n^\circ$ for all $x \geq 1$, we can simply use semidefinite programming. To verify that $W_1$ is indeed copositive, we compute
\begin{align}
    \v{^\top}W_1\v & = \frac{\displaystyle\sum_{j=1}^5 v_j\big(v_j - v_{j+1} + v_{j+2} + v_{j+3} - v_{j+4}\big)^2 + 4\sum_{j=1}^5v_jv_{j+1}v_{j+3}}{v_1 + v_2 + v_3 + v_4 + v_5} \geq 0
\end{align} 
whenever $\0 \neq \v \geq 0$, where the subscripts above are taken modulo $5$. Copositivity of $W_x$ when $x > 1$ then also follows by just noting that $W_x$ is the sum of $W_1$ and an entrywise non-negative matrix. More generally, to show that a matrix is (or is not) copositive or completely positive, we can make use of semidefinite programming hierarchies like the one introduced in \cite{Par00}.


\section{Free operations: Completely positive-preserving and completely positive completely positive maps}\label{sec:free_operations} 
  
The free operations in this resource theory are the quantum channels $\Phi \in \Lin(M_n,M_m)$ 
that preserve complete positivity. We refer to such channels as \emph{CP-preserving}, and they satisfy
\begin{equation} 
\Phi(X) \in \CP_m \; \text{ whenever } \; X \in \CP_n. 
\end{equation}

It is natural to ask which CP-preserving channels remain CP-preserving upon tensoring them with an identity channel of arbitrary size (i.e., when they act on just part of a quantum state rather than the whole state). For this reason, we call a linear map $\Phi \in \Lin(M_n,M_m)$ a \emph{completely positive completely positive (CPCP) map} if it satisfies 
\begin{equation} 
(\I_k \otimes \Phi)(X) \in \CP_{km} \; \text{ whenever } \; X \in \CP_{kn}. 
\end{equation}
Equivalently, a linear map is CPCP if $\I_k \otimes \Phi$ is CP-preserving for all $k \geq 1$.  
If we instead regard CPCP channels as the free operations then, in the terminology of \cite{CG18}, it gives this resource theory a tensor product structure.

The first result of this section characterizes CPCP maps in a few other ways that are analogous to the various well-known characterizations of completely positive maps.

\begin{theorem}\label{thm:CPCP_choi}
    Suppose $\Phi \in \Lin(M_n,M_m)$. The following are equivalent:
    \begin{enumerate}
        \item[(a)] $\Phi$ is CPCP.
        
        \item[(b)] $J(\Phi)$ is completely positive (as a matrix). 
        
        \item[(c)] $\Phi$ is completely positive (as a linear map) and has a family of real entrywise non-negative Kraus operators.
    \end{enumerate}
    Furthermore, the CP-rank of $J(\Phi)$ is the minimal number of entrywise non-negative Kraus operators possible in~(c).
\end{theorem}

\begin{proof}
    The fact that (a) implies (b) follows simply from the fact that $\sum_{i,j=0}^{n-1}\ketbra{i}{j} \otimes \ketbra{i}{j}$ is completely positive, so $J(\Phi) = (\I_n \otimes \Phi)\big(\sum_{i,j=0}^{n-1}\ketbra{i}{j} \otimes \ketbra{i}{j}\big)$ is also completely positive. 
    
    The fact that (b) implies (c) can be seen by noting that if $J(\Phi)$ is CP then we can write $J(\Phi) = \sum_{j=1}^m \mathbf{v_j}\mathbf{v}_{\mathbf{j}}^*$ for some entrywise non-negative column vectors $\{\mathbf{v_1}, \ldots, \mathbf{v_m}\}$. Then $\Phi(X) = \sum_{j=1}^m A_jXA_j^*$, where $A_j = \mathrm{mat}(\mathbf{v_j})$ is the column-by-column matricization of $\mathbf{v_j}$. Since each $\mathbf{v_j}$ is entrywise non-negative, so is each $A_j$. The ``furthermore'' part of the theorem can be seen by reversing the previous argument: if $m$ is the least integer for which we can write $\Phi(X) = \sum_{j=1}^m A_jXA_j^*$ with each $A_j$ entrywise non-negative, then letting $\mathbf{v_j} = \mathrm{vec}(A_j)$ be their vectorizations shows that $J(\Phi) = \sum_{j=1}^m \mathbf{v_j}\mathbf{v}_{\mathbf{j}}^*$ has CP-rank equal to $m$.
    
    Finally, to see that (c) implies (a) just notice that if $X = \sum_{i=1}^\ell \mathbf{v_i}\mathbf{v_i}^*$ is completely positive (with $\mathbf{v_i} \geq 0$ for all $i$) and $\{ A_1, \ldots, A_m \}$ are the entrywise non-negative Kraus operators of $\Phi$, then $(I_k \otimes A_j)\mathbf{v_i}$ is also entrywise non-negative, so
    \begin{equation}
        (\I_k \otimes \Phi)(X) = \sum_{j=1}^m (I_k \otimes A_j) \left(\sum_{i=1}^\ell \mathbf{v_i}\mathbf{v_i}^*\right) (I_k \otimes A_j)^* = \sum_{j=1}^m \sum_{i=1}^\ell \big((I_k \otimes A_j) \mathbf{v_i}\big)\big((I_k \otimes A_j) \mathbf{v_i}\big)^*
    \end{equation}
    is also completely positive for all $k \in \N$.
\end{proof} 

The above characterization of the Kraus operators of CPCP maps shows that CPCP maps are exactly the same as the maps that are positively factorizable via an abelian algebra, as concurrently investigated in \cite[Theorem~3.2]{LR21}. It also immediately gives us the following corollary.

\begin{corollary} \label{convex_composition} 
    The set $\CPCP$ is a convex cone and is closed under composition (when the composition is well-defined). 
\end{corollary} 

\begin{proof} 
The fact that $\CPCP$ is a convex cone follows immediately using 
Theorem~\ref{thm:CPCP_choi}$(c)$. To see that it is closed under composition, notice that if we have a CPCP map $\Phi_1$ with Kraus operators $\{ A_1, \ldots, A_m \}$ (where each have entrywise non-negative entries) and another CPCP map $\Phi_2$ with Kraus operators $\{ B_1, \ldots, B_\ell \}$ (where each has entrywise non-negative entries), then the composition $\Phi_2 \circ \Phi_1$ has Kraus operators 
\begin{equation} 
\{ B_j A_k : j \in \{ 1, \ldots, \ell \}, k \in \{ 1, \ldots, m \} \} 
\end{equation} 
and each of these clearly have non-negative entries as well, so $\Phi_2 \circ \Phi_1$ is thus CPCP.  
\end{proof}

We also have the following corollary concerning CPCP \emph{quantum channels} (i.e., members of $\CPCP$ that are trace-preserving), which follows via the same argument used to prove Corollary~\ref{convex_composition}.

\begin{corollary} \label{convex_composition_channel} 
    The set of CPCP channels is convex and closed under composition (when the composition is well-defined).
\end{corollary}

While we are not aware of a simple characterization of CP-preserving maps that is analogous to the one for CPCP maps above, we can at least see that these two sets are different (in all dimensions) by noting that if $\{\ket{v_0},\ket{v_1},\ldots,\ket{v_{n-1}}\}$ is any orthonormal basis of $\C^n$ other than the standard basis, then at least one entry of at least one of $\ketbra{v_0}{v_0}$, $\ketbra{v_1}{v_1}$, $\ldots$, $\ketbra{v_{n-1}}{v_{n-1}}$ must \emph{not} be real and non-negative (after all, they add up to the identity matrix, which has all off-diagonal entries equal to $0$). It follows that the measure-and-prepare channel $\Phi \in \Lin(M_n,M_n)$ defined by
\begin{equation}
    \Phi(X) = \sum_{j=0}^{n-1}\big(\bra{v_j}X\ket{v_j}\big) \ketbra{j}{j}
\end{equation} 
is CP-preserving (after all, if $X \succeq 0$ then $\bra{v_j}X\ket{v_j} \geq 0$ for all $j \in \{ 0, \ldots, n-1 \}$, so $\Phi(X)$ is completely positive), but is not CPCP (its Choi matrix is $J(\Phi) = \sum_{j=0}^{n-1} \overline{\ketbra{v_j}{v_j}} \otimes \ketbra{j}{j}$, which does not have all non-negative real entries, so $J(\Phi)$ is not completely positive).   

If we consider only CPCP \emph{quantum channels} (i.e., we add in the trace-preservation requirement) then CPCP maps simplify considerably. 

\begin{theorem}\label{thm:CPCP_channel_kraus}
    Suppose $\Phi \in \Lin(M_n,M_m)$ is a quantum channel. The following are equivalent:
    \begin{enumerate}
        \item[(a)] $\Phi$ is CPCP.
        
        \item[(b)] $\Phi$ has a family of entrywise non-negative Kraus operators with at most $1$ non-zero entry in each row.
    \end{enumerate}
\end{theorem}

\begin{proof}
    Theorem~\ref{thm:CPCP_choi} gives us all parts of this theorem except for the fact that the entrywise non-negative Kraus operators of $\Phi$ have at most one non-zero entry in each row. To see why this is the case, note that if $\Phi$ is trace-preserving then its Kraus operators $\{A_k\}$ satisfy $\sum_k A_k^*A_k = I$. If we denote the $i$-th column of $A_k$ by $\mathbf{a}_{i,k}$ then this is equivalent to
    \begin{equation} 
        \sum_k \mathbf{a}_{i,k} \cdot \mathbf{a}_{j,k} = \begin{cases}
            1 \ \ \text{if} \ i = j, \\
            0 \ \ \text{otherwise}.
        \end{cases}
    \end{equation} 
    Since $\mathbf{a}_{i,k},\mathbf{a}_{j,k} \geq 0$ for all $i$, $j$, and $k$, this implies $\mathbf{a}_{i,k} \cdot \mathbf{a}_{j,k} = 0$ whenever $i \neq j$. By again using the fact that $\mathbf{a}_{i,k},\mathbf{a}_{j,k} \geq 0$, this then implies that, for each $\ell$, either the $\ell$-th entry of $\mathbf{a}_{i,k}$ equals $0$ or the $\ell$-th entry of $\mathbf{a}_{j,k}$ equals $0$. Since $i$ and $j$ were arbitrary, this simply means that the $\ell$-th row of $A_k$ contains at most one non-zero entry, as claimed.
\end{proof}

When a CPCP quantum channel is furthermore unital (i.e., has $\Phi(I) = I$), it has an even simpler form that can be expressed in terms of \emph{Schur maps}, which are maps of the form 
\begin{equation} 
    \Phi_A(\rho) = A \odot \rho, 
\end{equation} 
where ``$\odot$'' denotes entrywise multiplication. Such a map $\Phi_A$ is completely positive if and only if $A \succeq 0$ \cite{Pau03}, and it is trace-preserving if and only if it is unital if and only if the diagonal entries of $A$ all equal $1$. We also say that $\Psi \in \Lin(M_n,M_n)$ is a \emph{permutation channel} if it can be written in the form 
\begin{equation} 
\Psi(\rho) = P\rho P^*, 
\end{equation} 
where $P$ is a permutation matrix. Equivalently, there is a permutation $\sigma$ such that, for all $i$ and $j$, the $(i,j)$-entry of $\Psi(\rho)$ equals $\rho_{\sigma(i),\sigma(j)}$.

\begin{theorem}\label{thm:CPCP_unital_channel}
    Suppose $\Phi \in \Lin(M_n,M_n)$ is a unital quantum channel. The following are equivalent:
    \begin{enumerate}
        \item[(a)] $\Phi$ is CPCP.
        
        \item[(b)] $\Phi$ has a family of entrywise non-negative Kraus operators with at most $1$ non-zero entry in each row and in each column.
        
        \item[(c)] $\Phi$ is a convex combination of maps of the form $\Psi \circ \Phi_A$, where $\Phi_A$ is a Schur channel, $\Psi$ is a permutation channel, and $A$ is completely positive.
    \end{enumerate}
\end{theorem}

\begin{proof} 
    The equivalence of (a) and (b) follows almost immediately from Theorem~\ref{thm:CPCP_channel_kraus}: we showed in the proof of that theorem that if the entrywise non-negative Kraus operators $\{A_k\}$ of $\Phi$ satisfy $\sum_k A_k^*A_k = I$ then they each have at most one non-zero entry in each row, and a similar argument shows that if $\Phi$ is unital then $\sum_k A_kA_k^* = I$, so each $A_k$ has at most one non-zero entry in each \emph{column} as well.
    
    To see that (b) and (c) are equivalent, notice that part~(b) is equivalent to saying that $\Phi$ has a family of Kraus operators of the form $A_k = P_k D_k$, where $P_k$ is a permutation matrix and $D_k$ is an entrywise non-negative diagonal matrix. If $\mathbf{d_k} = \mathrm{diag}(D_k)$ then it is straightforward to check that 
    \begin{equation}
    D_k\rho D_k^* = \big(\mathbf{d_k}\mathbf{d_k}^*\big) \odot \rho, 
    \end{equation} 
    where $\mathbf{d_k}\mathbf{d_k}^*$ is completely positive, and this argument can be reversed by writing $A$ as a convex combination of rank-$1$ non-negative matrices like $\mathbf{d_k}\mathbf{d_k}^*$.
\end{proof}

In particular, the above result shows that every unital CPCP quantum channel is a ``strictly incoherent operation'' \cite{YMG16}. Such channels are one of the natural choices of free operations in the resource theory of coherence \cite{CG16}.


\subsection[Examples of CPCP channels]{Examples of CPCP channels} 

In this subsection, we present several popular types of channels that are CPCP.

\begin{itemize} 
\item Identity, Pauli-$X$, and classical error channels: The identity channel is 
\begin{equation} 
\Phi(\rho) = \rho,   
\end{equation}  
the bit-flip channel, or Pauli-$X$  channel, is 
\begin{equation} 
\Phi(\rho) = X\rho X, 
\end{equation}  
where $X$ is the Pauli X matrix (see the upcoming Equation~\eqref{eq:pauli_mat}). The classical error channel is its convex combination 
\begin{equation} 
\Phi(\rho) = p\rho + (1-p)X\rho X \, \text{ for } \, p \in [0,1],   
\end{equation} 
which is CPCP since the set of CPCP maps is convex. 

\item Measure (in computational basis)-and-prepare (a CP state) channel: 
\begin{equation} 
\Phi(\rho) = \sum_{j=0}^{m-1} (\bra{j} \rho \ket{j})\sigma_j,  
\end{equation}  
where $\sigma_0, \ldots, \sigma_{m-1} \in \CP_n$.

\item 
The partially dephasing channel: 
\begin{equation} 
\Phi(\rho) = p\rho + (1-p)\tr(\rho) \, \frac{I_n}{n} \, \text{ for } \, p \in [0,1],   
\end{equation} 
which is CPCP by convexity.

\item
The partial trace: 
\begin{equation} 
\Phi(\rho) = \tr_1(\rho)   
\end{equation} 
since it has family of entrywise nonnegative Kraus operators with at most $1$ non-zero entry in each row.

\item Tensor/prepare channels: For a fixed $\sigma \in \CP$, 
\begin{equation} 
\Phi(\rho) = \rho \otimes \sigma.  
\end{equation}  
This is because the identity channel is CPCP and $\CP_m \otimes \CP_n \subseteq \CP_{mn}$. 
To see this, notice that $\mathbf{x}\mathbf{x}^{\top} \otimes \mathbf{y}\mathbf{y}^{\top} = (\mathbf{x} \otimes \mathbf{y}) (\mathbf{x} \otimes \mathbf{y})^{\top}$.

\item Stochastic, permutation, and SWAP channels: 
\begin{equation} 
\Phi(\rho) = S \rho S^{\top},   
\end{equation} 
where $S$ is a stochastic matrix (which includes permutations and doubly stochastic matrices as special cases).  
The SWAP channel: This channel is defined on product states as 
\begin{equation} 
\Phi(\rho \otimes \sigma) = \sigma \otimes \rho  
\end{equation} 
and is extended linearly.  
As this is a special case of a permutation channel, it is also CPCP. 

\item Schur maps and fully decohering channels: For $A \in \CP$ with diagonal entries equal to $1$, 
\begin{equation} 
\Phi_A(\rho) = A \odot \rho.
\end{equation} 
This case is detailed earlier in Theorem~\ref{thm:CPCP_unital_channel} and the preceding discussion. 
A special case is the fully decohering map 
\begin{equation} 
\Phi_{I}(\rho) = \mathrm{Diag}(\rho), 
\end{equation} 
where $\mathrm{Diag}(\rho)$ zeroes off the off-diagonal entries of $\rho$ and leaves the diagonal entries alone. 

\item Projections onto symmetric subspaces: The projection onto the symmetric subspace of $n \times n$ Hermitian matrices is defined as  
\begin{equation} 
S_n = \frac{1}{2} \I \otimes \I + \frac{1}{2} 
\left(    
\sum_{j,k = 0}^{n-1} \ketbra{k}{l} \otimes \ketbra{l}{k} 
\right).  
\end{equation} 
Thus, the quantum \emph{subchannel}\footnote{Note this map is not trace-preserving, but is nonetheless interesting in the study of quantum information.} which projects onto the symmetric subspace 
\begin{equation} 
\Phi(\rho) = S_n \rho S_n    
\end{equation}  
is CPCP (since it has a single Kraus operator which is entrywise non-negative).  
\end{itemize} 

As mentioned in Corollary~\ref{convex_composition_channel}, 
one can take the convex combination and/or compositions of any of these channels to create other CPCP channels.


\subsection{Doubly non-negative maps}

It is difficult to determine whether or not a linear map is CPCP (since it is NP-hard to determine whether or not its Choi matrix is CP), so it may be useful to instead consider the maps $\Phi \in \Lin(M_n,M_m)$ with the property that $(\I_k \otimes \Phi)(\rho) \in \DNN$ whenever $\rho \in \
\DNN$,  for all $k \in \N$. 
We call such a map \emph{completely positive doubly non-negative}, and denote the set of such maps as $\CPDNN$. The following theorem (which is directly analogous to Theorem~\ref{thm:CPCP_choi} for CPCP maps) tells us that it is simple to determine whether or not a map is $\CPDNN$. 

\begin{theorem}\label{thm:DNN_choi}
    Suppose $\Phi \in \Lin(M_n,M_m)$. The following are equivalent:
    
    \begin{enumerate}
        \item[(a)] $\Phi$ is $\CPDNN$. 
          
        \item[(b)] $J(\Phi)$ is doubly non-negative (as a matrix). 
    \end{enumerate}
\end{theorem}

\begin{proof}
    The fact that (a) implies (b) follows simply from the fact that $\sum_{i,j=0}^{n-1}\ketbra{i}{j} \otimes \ketbra{i}{j}$ is doubly non-negative, so $J(\Phi) = (I_n \otimes \Phi)\big(\sum_{i,j=0}^{n-1}\ketbra{i}{j} \otimes \ketbra{i}{j}\big)$ is doubly non-negative as well. For the reverse implication, recall the formula $\Phi(\rho) = \tr_1\big((\rho{^\top} \otimes I_m)J(\Phi)\big)$. If $J(\Phi)$ and $\rho$ are both doubly non-negative then it follows immediately that $\Phi(\rho)$ is as well. 
\end{proof}

If $n = m = 2$ then the set of CPDNN maps coincides with the set of CPCP maps simply because their Choi matrices are $4 \times 4$ in this case, and $\CP_4 = \DNN_4$. However, in all other cases there are maps that are CPDNN but not CPCP, as first demonstrated by the quantum channel ${\Phi \in \Lin(M_2,M_3)}$ with Choi matrix
\begin{equation}
    J(\Phi) = \frac{1}{6}\left[\begin{array}{ccc|ccc}
        3 & 0 & 0 & 0 & 0 & 0 \\
        0 & 1 & 1 & 0 & 0 & 1 \\
        0 & 1 & 2 & 1 & 0 & 0 \\\hline
        0 & 0 & 1 & 2 & 1 & 0 \\
        0 & 0 & 0 & 1 & 1 & 1 \\
        0 & 1 & 0 & 0 & 1 & 3
    \end{array}\right],
\end{equation} 
which is doubly non-negative but not completely positive. 
Indeed, the bottom-right $5 \times 5$ submatrix of this Choi matrix is the same (up to scaling) as the density matrix~\eqref{eq:dnn_not_cp}, and is thus DNN but not CP for the same reasons. In particular, this means that this channel does \emph{not} have a family of real entrywise non-negative Kraus operators. 
Note that one can add diagonal blocks with $\ket{0}\bra{0}$ to construct such examples of channels with a larger input dimension and pad each block with rows and columns of $0$s if one were to increase the output dimension. 
The only remaining case is for $\Lin(M_3, M_2)$. 
For this, consider the same Choi matrix above, but with a different partitioning:
\begin{equation}
    J(\Phi) = \frac{1}{6}\left[\begin{array}{cc|cc|cc}
        3 & 0 & 0 & 0 & 0 & 0 \\
        0 & 1 & 1 & 0 & 0 & 1 \\\hline
        0 & 1 & 2 & 1 & 0 & 0 \\
        0 & 0 & 1 & 2 & 1 & 0 \\\hline
        0 & 0 & 0 & 1 & 1 & 1 \\
        0 & 1 & 0 & 0 & 1 & 3
    \end{array}\right]. 
\end{equation} 
This Choi matrix is DNN but not CP as previously discussed. Note, however, that this is not a quantum channel as it does not satisfy trace-preservation. We leave it as a open question whether one can find a CPDNN quantum channel that is not CPCP for this specific choice of input and output dimension. Indeed, the difficulty is that there is no way to permute the rows and columns of this $J(\Phi)$ so as to create a trace-preserving map, and there are only a few other known choices for matrices in the set $\DNN \setminus \CP$.


\subsection{Maps on qubits} \label{sec:qubit_maps}
  
The structure of CP-preserving and CPCP maps simplifies quite a bit in the qubit-input qubit-output  (i.e., $m = n = 2$) case. We call such maps \emph{qubit channels}, for convenience.
Part of the reason for this is that, as noted earlier, they are simply the $\CPDNN$ maps in this case. Much more is known about this set in these small dimensions too. For example, the maximal CP-rank of a $4 \times 4$ matrix is $4$, so every {such} CPCP map has a set of $4$ (or fewer) entrywise non-negative Kraus operators.

To help illuminate the structure of the sets of CP-preserving and CPCP channels acting on qubits even further, recall \cite{NC00} that every density matrix $\rho \in \D_2$ can be written in the form
\begin{equation}
    \rho = \frac{1}{2}\big(I + \rho_x X + \rho_y Y + \rho_z Z\big),
\end{equation} 
where
\begin{equation}\label{eq:pauli_mat}
    X = \begin{bmatrix}
        0 & 1 \\ 1 & 0
    \end{bmatrix}, \quad Y =  \begin{bmatrix}
        0 & -i \\ i & 0
    \end{bmatrix}, \quad \text{and} \quad Z =  \begin{bmatrix}
        1 & 0 \\ 0 & -1
    \end{bmatrix}
\end{equation} 
are the Pauli matrices and $\rho_x$, $\rho_y$, and $\rho_z$ are the corresponding (real) coefficients of $\rho$ in the Pauli basis $\{I,X,Y,Z\}$. Positive semidefiniteness of $\rho$ is equivalent to $\rho_x^2 + \rho_y^2 + \rho_z^2 \leq 1$, so when written in this way, the set of qubit density matrices naturally form a ball of radius at most~$1$, called the \emph{Bloch ball}. 
It is straightforward to see that $\rho \in \CP_2 = \DNN_2$ if and only if we further impose $\rho_y = 0$ and $\rho_x \geq 0$, so $\CP_2$ makes up the $2$-dimensional ``wedge'' of the Bloch ball containing the $z$-axis and the positive half of the $x$-axis, as shown in Figure~\ref{fig:bloch_sphere}.

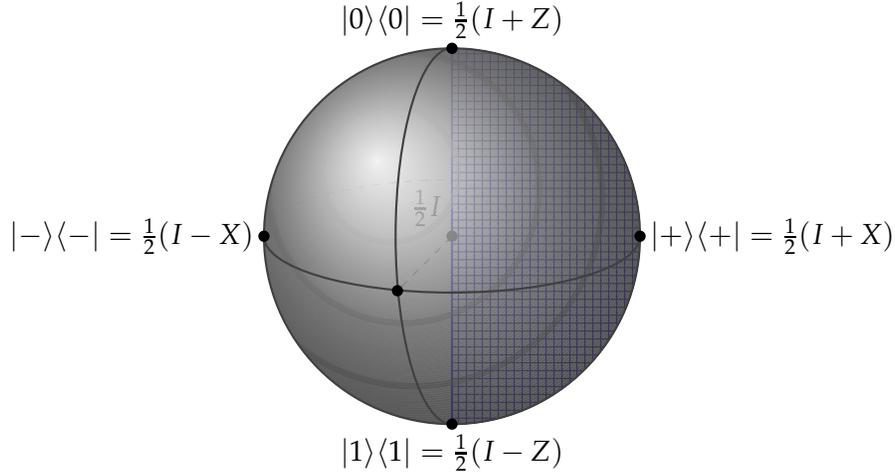
\begin{figure}[!htb]
	\centering
	\begin{tikzpicture}[scale=2.5]
    	\draw[black!75!white,dashed,opacity=0.5] (-1,0) arc (180:0:1cm and 0.3cm);
    	\draw[black!75!white,dashed,opacity=0.5] (0,1) arc (90:-90:0.3cm and 1cm);
    	
    	\draw[black!75!white,dashed] (0,0) -- (0,1);
    	\draw[black!75!white,dashed] (0,0) -- (-0.29,-0.29);
    	
        \begin{scope}
            \clip(0,-1) rectangle (1,1);
            \fill[white,opacity=0.5](0,0) circle (1cm);
            \draw[pattern=horizontal lines, pattern color=blue](0,0) circle (1cm);
            \draw[pattern=vertical lines, pattern color=blue](0,0) circle (1cm);
        \end{scope}
    	\draw[blue,thick] (0,1) -- (0,-1);
    	\fill (0,0) node[anchor=south east]{$\frac{1}{2}I$} circle (0.03cm);
    	
    	\shade[ball color=white!60!black,opacity=0.5] (0,0) circle (1cm);
    	
    	\draw[black!75!white,thick] (-1,0) arc (180:360:1cm and 0.3cm);
    	\draw[black!75!white,thick] (0,1) arc (90:270:0.3cm and 1cm);
    	\draw[black!75!white,thick] (0,0) circle (1cm);
    	
    	\fill (1,0) node[anchor=west]{$\ketbra{+}{+} = \frac{1}{2}(I+X)$} circle (0.03cm);
    	\fill (-1,0) node[anchor=east]{$\ketbra{-}{-} = \frac{1}{2}(I-X)$} circle (0.03cm);
    	\fill (-0.29,-0.29) circle (0.03cm);
    	\fill (0,1) node[anchor=south]{$\ketbra{0}{0} = \frac{1}{2}(I+Z)$} circle (0.03cm);
    	\fill (0,-1) node[anchor=north]{$\ketbra{1}{1} = \frac{1}{2}(I-Z)$} circle (0.03cm);
	\end{tikzpicture}
	
	\caption{The set of completely positive qubit density matrices forms the 2D ``wedge'' of the Bloch ball (cross-hashed on the right here) containing $\ketbra{0}{0}$, $\ketbra{1}{1}$, and $\ketbra{+}{+}$.}\label{fig:bloch_sphere}
\end{figure}

Since every quantum channel is a linear transformation acting on the vector space $M_2$, we can represent it as a matrix with respect to any basis of $M_2$ of our choosing. If we represent it with respect to the Pauli basis $\{I,X,Y,Z\}$ then its standard matrix has the form
\begin{align}\label{eq:qubit_map_representation}
    [\Phi] = \begin{bmatrix}
        1 & 0 & 0 & 0 \\
        t_x & T_{x,x} & T_{x,y} & T_{x,z} \\
        t_y & T_{y,x} & T_{y,y} & T_{y,z} \\
        t_z & T_{z,x} & T_{z,y} & T_{z,z}
    \end{bmatrix},
\end{align}
where the entries in this matrix keep track of how $\Phi$ acts on the different Pauli matrices. For example, $\Phi(I) = I + t_xX + t_yY + t_zZ$ and $\Phi(X) = T_{x,x}X + T_{y,x}Y + T_{z,x}Z$, and the coefficients of $\Phi(Y)$ and $\Phi(Z)$ similarly come from the third and fourth columns, respectively, of $[\Phi]$. The fact that $\{I,X,Y,Z\}$ is a basis of $M_2$ guarantees that $[\Phi]$ completely determines $\Phi$, and the special form of the first row of $[\Phi]$ comes from trace-preservation of $\Phi$ (recall that each of $X$, $Y$, and $Z$ are traceless, so they must be mapped to traceless matrices). Furthermore, since $\{I,X,Y,Z\}$ is a basis of the \emph{real} vector space of $2 \times 2$ Hermitian matrices, the entries of $[\Phi]$ are necessarily real whenever $\Phi$ is Hermiticity-preserving (which will always be the case for us).

The following theorem characterizes what CP-preserving and CPCP qubit channels look like when represented in this way.

\begin{theorem}\label{thm:CPCP_characterize_qubit} 
    Suppose $\Phi \in \Lin(M_2,M_2)$ is a quantum channel with standard matrix $[\Phi]$ as in~\eqref{eq:qubit_map_representation}. Then
    \begin{enumerate}
        \item[(a)] $\Phi$ is CP-preserving if and only if
        \begin{equation}
            [\Phi] = \begin{bmatrix}
                1 & 0 & 0 & 0 \\
                t_x & T_{x,x} & T_{x,y} & T_{x,z} \\
                0 & 0 & T_{y,y} & 0 \\
                t_z & T_{z,x} & T_{z,y} & T_{z,z}
            \end{bmatrix},
        \end{equation}
        where $t_x \geq |T_{x,z}|$ and $T_{x,x} \geq -\sqrt{t_x^2-T_{x,z}^2}$, and
        
        \item[(b)] $\Phi$ is CPCP if and only if it is CPDNN, if and only if
        \begin{equation}
            [\Phi] = \begin{bmatrix}
                1 & 0 & 0 & 0 \\
                t_x & T_{x,x} & 0 & T_{x,z} \\
                0 & 0 & T_{y,y} & 0 \\
                t_z & 0 & 0 & T_{z,z}
            \end{bmatrix},
        \end{equation} 
        where $t_x \geq |T_{x,z}|$, $T_{x,x} \geq |T_{y,y}|$, and $t_z \geq |T_{z,z}|-1$.
    \end{enumerate}
\end{theorem}

\begin{proof}
    For part~(a), we can see that $\Phi$ being CP-preserving implies the indicated restrictions on $[\Phi]$ by plugging certain specific free states $\rho \in \CP_2$ into $\Phi$. In particular, if $\rho = I/2$ then 
    \begin{equation} 
    \Phi(\rho) = (I + t_x X + t_y Y + t_z Z)/2 \in \CP_2, 
    \end{equation} 
    so $t_y = 0$ and $t_x \geq 0$. If $\rho = (I \pm Z)/2$ then
    \begin{equation}
        \Phi(\rho) = \frac{1}{2}\big(I + (t_x \pm T_{x,z})X \pm T_{y,z}Y + (t_z \pm T_{z,z})Z\big) \in \CP_2,
    \end{equation} 
    so $T_{y,z} = 0$ and $t_x \pm T_{x,z} \geq 0$ (i.e., $|T_{x,z}| \leq t_x$). Finally, if $\rho = (I + \rho_xX + \rho_zZ)/2$, where ${\rho_z = -T_{x,z}/t_x}$ (if $t_x = 0$ and/or $|T_{x,z}| = t_x$ then choose $\rho = (I + \rho_x X)/2$ with $\rho_x > 0$ arbitrary instead to avoid division by $0$ and make this argument work) and $\rho_x = \sqrt{1 - \rho_z^2}$ then
    \begin{equation}
        \Phi(\rho) = \frac{1}{2}\big(I + (t_x + \rho_x T_{x,x} + \rho_z T_{z,z})X + \rho_x T_{y,x}Y + (t_x + \rho_x T_{z,x} + \rho_z T_{z,z})Z\big) \in \CP_2,
    \end{equation}
    so $T_{y,x} = 0$ and the coefficient of $X$ in $\Phi(\rho)$ is
    \begin{equation}
        t_x + \rho_x T_{x,x} + \rho_z T_{z,z} = t_x + \sqrt{1 - \frac{T_{x,z}^2}{t_x^2}}T_{x,x} - \frac{T_{x,z}}{t_x}T_{x,z},
    \end{equation}
    which we can see (by multiplying through by $t_x$) is non-negative if and only if 
    \begin{equation}
    t_x^2 - T_{x,z}^2 + \sqrt{t_x^2 - T_{x,z}^2}T_{x,x} \geq 0, 
    \end{equation}
    which is equivalent to $T_{x,x} \geq -\sqrt{t_x^2 - T_{x,z}^2}$, as claimed.
    
    In the other direction, to see that $[\Phi]$ having the form described in part~(a) of the theorem implies that $\Phi$ is CP-preserving, we note that every $\rho \in \CP_2$ has $[\rho] = (1,\rho_x,0,\rho_z)$ for some $\rho_x \geq 0$ and $\rho_x^2 + \rho_z^2 \leq 1$. Then
    \begin{equation}
        [\Phi][\rho] = \begin{bmatrix}
                1 & 0 & 0 & 0 \\
                t_x & T_{x,x} & T_{x,y} & T_{x,z} \\
                0 & 0 & T_{y,y} & 0 \\
                t_z & T_{z,x} & T_{z,y} & T_{z,z}
        \end{bmatrix}\begin{bmatrix}
            1 \\
            \rho_x \\
            0 \\
            \rho_z
        \end{bmatrix} = \begin{bmatrix}
            1 \\
            t_x+T_{x,x}\rho_x + T_{x,z}\rho_z \\
            0 \\
            t_z+T_{z,x}\rho_x + T_{z,z}\rho_z
        \end{bmatrix}.
    \end{equation} 
    To see that $\Phi(\rho) \in \CP_2$ (and thus $\Phi$ is CP-preserving) we thus just need to check that
    \begin{align}\label{ineq:CP_pres_qubit_to_check}
        t_x+T_{x,x}\rho_x + T_{x,z}\rho_z \geq 0 \quad \text{whenever} \quad \rho_x \geq 0, \rho_x^2 + \rho_z^2 \leq 1.
    \end{align}
    To this end, just notice that if $t_x \geq 0$ and $|T_{x,z}| \leq t_x$ then the inequality $T_{x,x} \geq -\sqrt{t_x^2-T_{x,z}^2}$ is equivalent to $T_{x,x} \geq 0$ or $\sqrt{T_{x,x}^2 + T_{x,z}^2} \leq t_x$. Each of these inequalities straightforwardly imply Inequality~\eqref{ineq:CP_pres_qubit_to_check}, which completes the proof of part~(a).
    
    For part~(b), we note that the Choi matrix $J(\Phi)$ of $\Phi$ is $4 \times 4$ and thus completely positive if and only if it is doubly non-negative, so Theorems~\ref{thm:CPCP_choi} and~\ref{thm:DNN_choi} tell us that $\Phi$ is CPCP if and only if it is doubly non-negative. To determine the form of $[\Phi]$, we use Theorem~\ref{thm:DNN_choi} to see that $\Phi$ being doubly non-negative is equivalent to $J(\Phi)$ being doubly non-negative. By using trace-preservation of $\Phi$, we see that this is equivalent to $J(\Phi)$ having the form
    \begin{equation}
        J(\Phi) = \begin{bmatrix}
            a & b & 0 & c \\
            b & 1-a & d & 0 \\
            0 & d & e & f \\
            c & 0 & f & 1-e
        \end{bmatrix},
    \end{equation}
    where $a,b,c,d,e,f \geq 0$ and $a,e \leq 1$. Straightforward (but tedious and ugly) linear algebra shows that this is equivalent to the standard matrix of $\Phi$ having the form
    \begin{equation}
        [\Phi] = \begin{bmatrix}
            1 & 0 & 0 & 0 \\
            b+f & c+d & 0 & b-f \\
            0 & 0 & c-d & 0 \\
            a+e-1 & 0 & 0 & a-e
        \end{bmatrix}.
    \end{equation} 
    By making the change of variables $t_x = b+f$, $T_{x,z} = b-f$, $T_{x,x} = c+d$, $T_{y,y} = c-d$, and ${t_z = a+e-1}$, $T_{z,z} = a-e$, we see that non-negativity of $a,b,c,d,e$, and $f$ is equivalent to $t_x \geq 0$, $|T_{x,z}| \leq t_x$, $T_{x,x} \geq 0$, $|T_{y,y}| \leq T_{x,x}$, $t_z \geq -1$, and $|T_{z,z}| \leq t_z + 1$. We can discard the inequality $t_z \geq -1$ since it follows for free from complete positivity of $\Phi$ (if $t_z < -1$ then the $(1,1)$-entry of $\Phi(I)$ is $(1 + t_z)/2 < 0$), which completes the proof.
\end{proof}

If we restrict our attention slightly further to \emph{unital} CPCP qubit channels, we get the following simplification of the previous theorem.

\begin{corollary}
    Suppose $\Phi \in \Lin(M_2,M_2)$ is a unital quantum channel with standard matrix $[\Phi]$ as in Equation~\eqref{eq:qubit_map_representation}. Then $\Phi$ is CPCP if and only if it is CPDNN, if and only if
        \begin{equation}
            [\Phi] = \begin{bmatrix}
                1 & 0 & 0 & 0 \\
                0 & T_{x,x} & 0 & 0 \\
                0 & 0 & T_{y,y} & 0 \\
                0 & 0 & 0 & T_{z,z}
            \end{bmatrix}, 
        \end{equation} 
        where $T_{x,x} \geq |T_{y,y}|$. 
\end{corollary}

\begin{proof}
    $\Phi$ being unital is equivalent to its standard matrix~\eqref{eq:qubit_map_representation} having $t_x = t_y = t_z = 0$. When combined with Theorem~\ref{thm:CPCP_characterize_qubit}, the result follows immediately.
\end{proof}

In other words, the above corollary says that every unital CPCP qubit channel $\Phi$ acts on the Pauli $X$, $Y$, and $Z$ matrices independently:
\begin{equation}
    \Phi(I + \rho_xX + \rho_yY + \rho_zZ) = I + \rho_x T_{x,x}X + \rho_y T_{y,y}Y + \rho_z T_{z,z}Z,
\end{equation} 
where $T_{x,x}, T_{y,y}, T_{z,z} \in \R$ satisfy $|T_{y,y}| \leq T_{x,x} \leq 1$ and $|T_{z,z}| \leq 1$ (with the final two inequalities being equivalent to complete positivity of $\Phi$).

This representation of a unital CPCP qubit channel of course agrees with Theorem~\ref{thm:CPCP_unital_channel}, since if we use part~(c) of that theorem, we can write
\begin{equation}
    \Phi(\rho) = pA \odot \rho + (1-p)X (B \odot \rho) X,
\end{equation} 
where $A$ and $B$ are CP with diagonal entries equal to $1$ and $p \in [0,1]$ is a scalar. 
Then by just rewriting things appropriately, we can show that 
\begin{equation}
    \Phi(I + \rho_xX + \rho_yY + \rho_zZ) = I + \rho_x\big(pa_{1,2} + (1-p)b_{1,2}\big)X + \rho_y\big(pa_{1,2} - (1-p)b_{1,2}\big)Y + \rho_z(2p - 1)Z.
\end{equation} 
In particular, we have $T_{x,x} = pa_{1,2} + (1-p)b_{1,2}$, $T_{y,y} = pa_{1,2} - (1-p)b_{1,2}$ (which satisfies ${|T_{y,y}| \leq T_{x,x}}$), and $T_{z,z} = 2p-1$.

     
\subsection{The most resourceful state} 

We now ask whether there is a ``most resourceful state'', which in this context translates into determining which state is the ``most non-non-negative''. 
Note that we did not say ``most positive'', for reasons that follow from the below discussion. 
There are a few ways to quantify this; one way is to find a state which maps to \emph{any} other state using some definition of a free operation (this is analogous to how pure quantum states with all of their Schmidt coefficients equal to each other are considered ``maximally entangled'' since they can be converted to any other quantum state via LOCC operations).     
Since we have a nice characterization of the free operations for qubit channels, we can characterize the most resourceful qubit state, below. 

\begin{lemma} \label{lem:most_resourceful}
For any qubit state $\rho$, there exists a CP-preserving quantum channel $\Phi$ such that 
\begin{equation} 
\Phi(\sigma) = \rho 
\end{equation} 
where $\sigma$ is the density matrix
\begin{equation} 
    \sigma := \frac{I + Y}{2}.
\end{equation} 
\end{lemma}

\begin{proof} 
    Define a linear map $\Phi \in \Lin(M_2,M_2)$ which acts on the standard basis matrices in the following way:   
    \begin{align} 
        \Phi(\ketbra{0}{0}) & = \frac{I}{2} \\ 
        \Phi(\ketbra{0}{1}) & = i \, \frac{aX+bY+cZ}{2} \\ 
        \Phi(\ketbra{1}{0}) & = -i \, \frac{aX+bY+cZ}{2} \\ 
        \Phi(\ketbra{1}{1}) & = \frac{I}{2} 
    \end{align}  
    for real parameters $a,b,c$ satisfying $a^2 + b^2 + c^2 \leq 1$, and extend linearly.  
    Its Choi matrix is, in block form, 
    \begin{equation}
        \frac{1}{2} \begin{bmatrix}
            I & i(aX+bY+cZ) \\ 
            -i(aX+bY+cZ) & I
        \end{bmatrix},
    \end{equation} 
    which can be checked to be positive semidefinite using Schur complements. Thus, $\Phi$ is completely positive.  
    Representing this channel in the Pauli basis, we have 
    \begin{equation} 
    [\Phi] = \begin{bmatrix} 
    1 & 0 & 0 & 0 \\ 
    0 & 0 & a & 0 \\ 
    0 & 0 & b & 0 \\ 
    0 & 0 & c & 0 
    \end{bmatrix}
    \end{equation} 
    and thus it is also trace-preserving. 
    From Theorem~\ref{thm:CPCP_characterize_qubit}, we see that this channel is CP-preserving. 
    Since for fixed $a,b,c$, the channel satisfies
    \begin{equation} \label{eq:most_resourceful_Y}
    \Phi(\sigma) = \frac{I + aX + bY + cZ}{2},
    \end{equation} 
    and every qubit state is of this form with $a^2 + b^2 + c^2 \leq 1$, we see that we can choose $a,b,c$ such that this represents any qubit we want. 
\end{proof}

The above lemma says that the quantum state $\sigma = (I + Y)/2$ is maximally resourceful in this resource theory. In fact, the proof of the lemma shows that $\sigma$ is essentially unique---the only other maximally resourceful state is $(I-Y)/2$. To verify that $(I-Y)/2$ is also maximally resourceful, note that the only part of the proof of Lemma~\ref{lem:most_resourceful} that changes in this case is that $b$ changes to $-b$ in Equation~\eqref{eq:most_resourceful_Y}.

To see that no other states are maximally resourceful, we just note that complete positivity (in the linear map sense) forces the coefficient $T_{y,y}$ in Theorem~\ref{thm:CPCP_characterize_qubit}(a) to satisfy $|T_{y,y}| \leq 1$, since otherwise $\Phi(I \pm Y) = I + (t_x \pm T_{x,y})X \pm T_{y,y}Y + (t_z \pm T_{z,y})Z$ would not be positive semidefinite (e.g., if $\ket{v} = (1,i)^\top/\sqrt{2}$ then $\bra{v}\Phi(I \pm Y)\ket{v} = 1 \pm T_{y,y}$ is less than zero for one of the two choices of sign). Since Theorem~\ref{thm:CPCP_characterize_qubit}(a) tells us that a general qubit state $\rho = (I + \rho_x X + \rho_y Y + \rho_z Z)/2$ is such that
\begin{equation}
    \Phi(\rho) = \frac{1}{2}\big(I + (t_x + T_{x,x}\rho_x + T_{x,y}\rho_y + T_{x,z}\rho_z)X + T_{y,y}\rho_y Y + (t_z + T_{z,x}\rho_x + T_{z,y}\rho_y + T_{z,z}\rho_z)Z\big),
\end{equation}
we conclude that the only way that a qubit CP-preserving $\Phi$ can satisfy $\Phi(\rho) = (I+Y)/2$ is if $T_{y,y}\rho_y = \pm 1$, which forces $\rho_y = \pm 1$, so $\rho = (I \pm Y)/2$. 

This furthermore tells us that we cannot use CPCP quantum channels to map $\sigma = (I+Y)/2$ or any other state to arbitrary qubits. 
To see this, Theorem~\ref{thm:CPCP_characterize_qubit} says that if $\Phi$ is a CPCP quantum channel, then 
\begin{equation}  \Phi(\sigma) 
= \frac{I + t_xX + T_{y,y}Y + t_zZ}{2}, 
\end{equation} 
where the constraints on $t_x$, $T_{y,y}$, and $t_z$ are as in part (b) of that theorem. 
In particular, since $t_x \geq 0$, this tells us that we cannot get any qubits with a negative weight on the $X$ matrix, so $\sigma$ cannot be mapped to arbitrary qubits in this way (and a similar argument works for $(I-Y)/2$). Since we already showed that no other state can be mapped by CP-preserving channels to arbitrary qubits, they certainly cannot be mapped by CPCP channels to arbitrary qubits. We thus conclude that there is no maximally resourceful state in this resource theory if we consider only CPCP channels as the free operations.


\section{Measures of non-negativity}  \label{sec:measures_nonneg}

We now investigate some ways of quantifying how close to free (i.e., pure and non-negative, or mixed and completely positive) a non-free state is. That is, we define and investigate measures that are analogous in this resource theory to entanglement monotones \cite{Vid00} for the resource theory of entanglement and to coherence monotones \cite{BCP14} for the resource theory of coherence.

Throughout this section (and the remainder of this paper), we omit subscripts on sets like $\D_n$, $\CP_n$, and $\DNN_n$, and instead simply note that $n$ denotes the dimensionality of the states in question in all cases. We say that a function $N : \D \rightarrow [0,\infty]$ is a \textit{non-negativity monotone} if it satisfies the following three properties:

\begin{itemize}
    \item[(C1)] \textbf{Freeness:} $N(\rho) = 0$ whenever $\rho \in \CP$,

    \item[(C2)] \textbf{Monotonicity:} $N\big(\Phi(\rho)\big) \leq N(\rho)$ for all CP-preserving channels $\Phi$ and all $\rho \in \D$, and

    \item[(C3)] \textbf{Convexity:} $N(\sum_{i=1}^k p_i\rho_i) \leq \sum_{i=1}^k p_i N(\rho_i)$ whenever $\{ \rho_1, \ldots, \rho_k \} \subset \D$, $p_i \geq 0$ for all ${i \in \{ 1, \ldots, k \}}$, and $\sum_{i=1}^k p_i = 1$.
\end{itemize}
Optionally, either or both of the following properties may be enforced as well:
\begin{itemize}
    \item[(C1b)] \textbf{Faithfulness:} $N(\rho) = 0$ if and only if $\rho \in \CP$, and/or
    
    \item[(C2b)] \textbf{Strong monotonicity:} $\sum_{i=1}^k p_i N\big(\Phi_i(\rho)/p_i\big) \leq N(\rho)$, where $p_i = \tr\big(\Phi_i(\rho)\big)$ for all $i \in \{ 1, \ldots, k \}$, whenever each $\Phi_i$ is CP-preserving and $\sum_{i=1}^k \Phi_i$ is trace-preserving (i.e., a quantum channel).
\end{itemize}

We note that faithfulness (C1b) trivially implies freeness (C1). Similarly, strong monotonicity (C2b) trivially implies monotonicity (C2), and physically corresponds to the function $N$ being monotonic not just under the free quantum operations, but also under subchannels or measurements.
  
Since the set of CP-preserving maps is somewhat unwieldy, it is typically easier to check the monotonicity and strong monotonicity properties (C2) and (C2b) if we instead only consider CPCP channels. We note that in this case strong monotonicity (C2b) is equivalent (via Theorem~\ref{thm:CPCP_channel_kraus}) to the requirement that
\begin{equation}
    \sum_{i=1}^k p_i N\big(A_i\rho A_i^*)/p_i\big) \leq N(\rho),
\end{equation} 
where $p_i = \tr\big(A_i\rho A_i^*\big)$ for all $i \in \{ 1, \ldots, k \}$, whenever each $A_i$ is entrywise non-negative and also $\sum_{i=1}^k A_i^*A_i = I$ (and thus in particular has at most one non-zero entry in each row).
  
  
\subsection[The 1-norm of non-negativity for pure states]{The $1$-norm of non-negativity for pure states}\label{sec:l1_pure_measure}

Before investigating any proper non-negativity monotones, we first introduce and explore a non-standard vector norm that will be of use to us later. We will see that we can roughly think of this norm as providing a measure of non-negativity of pure states.

\begin{definition}\label{defn:phasey_norms_pure}
    Suppose $\v \in \C^n$. The $\mathbf{1}$\textbf{-norm of non-negativity} is the quantity 
    \begin{align}
        \|\mathbf{v}\|_1^{\textup{N}} & \defeq \inf\left\{ \sum_j \|\mathbf{v_j}\| : \mathbf{v} = \sum_j c_j\mathbf{v_j}, \ \mathbf{v_j} \geq 0, \ |c_j| = 1 \ \forall j \right\},
    \end{align}
    where the infimum is taken over all such finite decompositions of $\v$.
\end{definition}

While this quantity perhaps look quite strange at first, it is the natural analog of well-known quantities from the resource theories of coherence and entanglement. For example, if we replace the free states $\{\mathbf{v_j}\}$ in this definition by incoherent states (i.e., states with just one non-zero entry) then the resulting norm is just the usual $1$-norm $\|\mathbf{v}\|_1 = \sum_j |v_j|$. On the other hand, if we replace those free states by separable (pure) states then the resulting norm is the sum of Schmidt coefficients of $\v$.

We also note that it is straightforward to see that if $\ket{v} \in \C^n$ is a pure state (i.e., has $\|\ket{v}\| = 1$) then $\|\ket{v}\|_1^{\textup{N}} \geq 1$, and furthermore equality holds if and only if every entry of $\ket{v}$ has the same phase (i.e., if and only if $\ketbra{v}{v} \in \CP_n$). The following theorem establishes some less trivial bounds on this norm. 

\begin{theorem}\label{thm:nn_norm_bounds}
    If $\ket{v} \in \C^n$ is a pure state then $\|\ket{v}\|_1^{\textup{N}} \leq \min\{\sqrt{n},2\}$.
\end{theorem}

\begin{proof}
    The $\sqrt{n}$ upper bound follows from the bound $\|\ket{v}\|_1^{\textup{N}} \leq \|\ket{v}\|_1 \leq \sqrt{n}$. The (dimension-independent!) upper bound of $2$ follows from the fact that we can write
    \begin{equation}
        \ket{v} = \max(\Re(\ket{v}),0) + (-1)\min(\Re(\ket{v}),0) + i\max(\Im(\ket{v}),0) + (-i)\min(\Im(\ket{v}),0),
    \end{equation} 
    where $\Re(\ket{v})$ and $\Im(\ket{v})$ are the (entry-wise) real and imaginary parts of $\ket{v}$, respectively, and each maximization and minimization is also meant entry-wise. This is a decomposition of the type from Definition~\ref{defn:phasey_norms_pure}, so
    \begin{align}\begin{split}
        \|\ket{v}\|_1^{\textup{N}} & \leq \|\max(\Re(\ket{v}),0)\| + \|\min(\Re(\ket{v}),0)\| + \|\max(\Im(\ket{v}),0)\| + \|\min(\Im(\ket{v}),0)\| \\
        & \leq \sqrt{2}\big( \|\Re(\ket{v})\| + \|\Im(\ket{v})\| \big) \\
        & \leq 2\|\ket{v}\| = 2 
    \end{split}\end{align} 
    completing the proof.    
\end{proof}

When $n = 2$, we will see shortly (by combining the upcoming Theorems~\ref{thm:rob_on_pure} and~\ref{thm:rob_qubits}) that we have the explicit formula 
\begin{equation}
\|\ket{v}\|_1^{\textup{N}} = \sqrt{2\big|\min\{\Re(v_1\overline{v_2}),0\} + i\Im(v_1\overline{v_2})\big|+1}. 
\end{equation} 
In particular, this tells us that the bound of Theorem~\ref{thm:nn_norm_bounds} can be tight when $n = 2$, since $\|\ket{v}\|_1^{\textup{N}} = \sqrt{2}$ when $\ket{v} = \tfrac{1}{\sqrt{2}}(1,-1){^\top}$.

We will furthermore see that $\|\ket{v}\|_1^{\textup{N}}$ can be computed via semidefinite programming when $n \leq 4$, and this can be quickly used to show that the bound provided by Theorem~\ref{thm:nn_norm_bounds} can be tight in all dimensions. For example,
\begin{align}
    \Big\|\tfrac{1}{\sqrt{3}}(1,-1,i){^\top}\Big\|_1^{\textup{N}} = \sqrt{3} \quad \text{and} \quad \Big\|\tfrac{1}{\sqrt{n}}(1,-1,i,-i,0,\ldots,0){^\top}\Big\|_1^{\textup{N}} & = 2 \quad \text{for all} \quad n \geq 4. 
\end{align} 

In general, it is not clear that there is a simple way to compute $\|\ket{v}\|_1^{\textup{N}}$ via standard techniques like semidefinite programming, but we can approximate it very well in practice by making use of nets. In particular, if we let $k \geq 1$ be a large integer and choose the scalars $\{c_j\}$ in Definition~\ref{defn:phasey_norms_pure} to be equally spaced around the unit circle in the complex plane, then we can find the corresponding optimal vectors $\{\mathbf{v_j}\}$ via the following semidefinite program: 
\begin{align}\label{eq:sdp_1norm_net}\begin{split}
	\text{minimize:}\quad & \|\v_{\mathbf{0}}\| + \|\v_{\mathbf{1}}\| + \cdots + \|\v_{\mathbf{k-1}}\| \\
	\text{subject to:}\quad & \mathbf{v} = \v_{\mathbf{0}} + e^{2i\pi/k}\v_{\mathbf{1}} + e^{4i\pi/k}\v_{\mathbf{2}} + \cdots + e^{2(k-1)i\pi/k}\v_{\mathbf{k-1}},\\
		& \mathbf{v_j} \geq 0 \ \text{for all} \ 0 \leq j < k.
\end{split}\end{align} 

The following theorem provides a bound on the error of this semidefinite program. 

\begin{theorem}\label{thm:sdp_1norm}
    Suppose $k$ is a positive multiple of $4$ and $\v \in \C^n$. The optimal value $\alpha_k$ of the semidefinite program~\eqref{eq:sdp_1norm_net} satisfies
    \begin{equation}
        \frac{\alpha_k}{1 + 10\sin\big(\frac{\pi}{2k}\big)} \leq \|\mathbf{v}\|_1^{\textup{N}} \leq \alpha_k.
    \end{equation} 
\end{theorem}

\begin{proof}
    The inequality $\|\mathbf{v}\|_1^{\textup{N}} \leq \alpha_k$ comes from the fact that $\alpha_k$ arises from a particular decomposition of the type that we minimize over in the definition of $\|\mathbf{v}\|_1^{\textup{N}}$.
    
    For the other inequality, let $\alpha_k(\mathbf{v})$ denote the optimal value of the semidefinite program~\eqref{eq:sdp_1norm_net} when applied to the vector $\mathbf{v}$. We need two facts: (a) $\alpha_{k}(\mathbf{v} + \mathbf{w}) \leq \alpha_k(\mathbf{v}) + \alpha_k(\mathbf{w})$, which follows immediately from the triangle inequality for the usual Euclidean norm $\|\cdot\|$, and (b) $\alpha_k(c\mathbf{v}) \leq 2\big|c-|c|\big| \big\|\mathbf{v}\big\| + |c|\alpha_k(\mathbf{v})$, which we now demonstrate:
    \begin{align}
        \alpha_k(c\mathbf{v}) - |c|\alpha_k(\mathbf{v}) & \leq \alpha_k(c\mathbf{v} - |c|\mathbf{v}) \\
        & \leq 2\big\|c\mathbf{v} - |c|\mathbf{v}\big\| = 2\big|c-|c|\big| \big\|\mathbf{v}\big\|,
    \end{align}
    where the first inequality above comes from fact~(a), and the second inequality comes from the fact that $k$ is a multiple of $4$ so we can use the argument from the proof of Theorem~\ref{thm:nn_norm_bounds}. Rearranging gives $\alpha_k(c\mathbf{v}) \leq 2\big|c-|c|\big| \big\|\mathbf{v}\big\| + |c|\alpha_k(\mathbf{v})$, as desired.
    
    Now let $\varepsilon > 0$ be small and suppose that
    \begin{equation}
        \v = \sum_j c_j\mathbf{v_j}
    \end{equation} 
    is a decomposition which almost attains the infimum in Definition~\ref{defn:phasey_norms_pure}: $\sum_j \|\v_{\mathbf{j}}\| < \|\v\|_1^{\textup{N}} + \varepsilon$. Also let $d_j$ be the closest $k$-th root of unity to $c_j$ in the complex plane. Some straightforward geometry shows that the angle between $c_j$ and $d_j$ in the complex plane is no larger than $\pi/k$, so $|c_j - d_j| \leq 2\sin\big(\pi/(2k)\big)$. If we define $\w = \sum_j d_j\mathbf{v_j}$ then
    \begin{align}
        \alpha_k(\mathbf{w}) \leq \sum_j \|\mathbf{v_j}\| < \|\v\|_1^{\textup{N}} + \varepsilon,
    \end{align}
    and the fact that $\mathbf{v_j} \geq 0$ for all $j$ tells us that $\|\mathbf{v_j}\| = \|\mathbf{v_j}\|_1^{\textup{N}} = \alpha_k(\mathbf{v_j})$. This implies
    \begin{align}
        \alpha_k(\v-\w) & = \alpha_k\left(\sum_j (c_j-d_j)\mathbf{v_j}\right) \leq \sum_j \alpha_k\big((c_j-d_j)\mathbf{v_j}\big) \\
        & \leq \sum_j \Big(2\big|(c_j-d_j)-|c_j-d_j|\big| \big\|\mathbf{v_j}\big\| + |c_j-d_j|\alpha_k(\mathbf{v_j})\Big) \\
        & \leq 5\sum_j |c_j-d_j| \big\|\mathbf{v_j}\big\| \\
        & \leq 10\sum_j \sin\Big(\frac{\pi}{2k}\Big)\big\|\mathbf{v_j}\big\| < 10\sin\Big(\frac{\pi}{2k}\Big)\big(\|\mathbf{v}\|_1^{\textup{N}} + \varepsilon\big),
    \end{align}
    where the first inequality comes from property~(a) above, the second inequality comes from property~(b) above, and the third inequality comes from applying the triangle inequality to the absolute value and using the fact that $\|\mathbf{v_j}\| = \alpha_k(\mathbf{v_j})$.
    
    This inequality, together with the triangle inequality for $\alpha_k$ (i.e., property~(a) above), and letting $\varepsilon \rightarrow 0^{+}$, then shows that
    \begin{equation}
        \alpha_k(\mathbf{v}) = \alpha_k(\w + (\v-\w)) \leq \alpha_k(\mathbf{w}) + \alpha_k(\v-\w) \leq \Big(1 + 10\sin\Big(\frac{\pi}{2k}\Big)\Big)\|\mathbf{v}\|_1^{\textup{N}},
    \end{equation} 
    completing the proof.
\end{proof}

MATLAB code that implements the semidefinite program~\eqref{eq:sdp_1norm_net} and all bounds that we have seen for this norm is provided at \cite{SuppCode}, and in practice it can compute this norm to $4$ decimal places when $n = 50$ in about $10$ seconds on a standard desktop computer.


\subsection{The robustness of non-negativity}\label{sec:rob_nonneg}

We define the robustness of non-negativity as follows (in analogy with the robustnesses of entanglement \cite{VT99} and coherence \cite{NBCPJA16}):
\begin{align}
    N^{\textup{R}}_{\CP}(\rho) & \defeq \min_{\sigma\in \D}\left\{s\geq 0\,:\, \frac{\rho+s\sigma}{1+s}\in \CP\right\} \label{eq:cp_robustness}.
\end{align} 
We note that it follows from \cite[Theorems~14, 15, and 18]{Reg18} that $N^{\textup{R}}_{\CP}$ is a non-negativity monotone in the strongest possible sense: it satisfies properties (C1), (C2), and (C3) from earlier, as well as the stronger properties (C1b) and (C2b).

While this quantity can naturally be expressed as a conic optimization problem, optimizing over the set $\CP$ is NP-hard, so it is useful to be able to get explicitly computable bounds on it. For this reason, we note that duality theory for conic optimization (see \cite{BV04} for details) says that we can rewrite $N^{\textup{R}}_{\CP}(\rho)$ as the following optimization over the dual cone $\CP^\circ$:
\begin{align}\label{eq:cp_robustness_dual}
    N^{\textup{R}}_{\CP}(\rho) & = \max_{W\in \CP^\circ}\left\{-\tr(W\rho)\,:\, W \preceq I\right\}.
\end{align} 
We recall from Section~\ref{sec:copositive} that $\CP^\circ$ is the set of matrices whose entrywise real part is copositive.

This dual formulation of $N^{\textup{R}}_{\CP}$ is useful because we can use any copositive matrix (many of which are known in the literature) to get a lower bound on $N^{\textup{R}}_{\CP}(\rho)$. Furthermore, this lower bound is measurable since we can treat that copositive matrix as an observable that we measure on the state $\rho$, and the quantity $-\tr(W\rho)$ that lower bounds $N^{\textup{R}}_{\CP}(\rho)$ is simply the negative of the average value of this measurement. For example, we noted earlier that if $W_1$ is the Horn matrix~\eqref{eq:horn_like_Wx} and $\rho$ is the doubly non-negative but not completely positive density matrix~\eqref{eq:dnn_not_cp} then $\tr(W_1\rho) = -1/9$. If we set ${W = W_1/\lambda_{\textup{max}}(W_1) = W_1/(\sqrt{5}+1)}$ so that $W \preceq I$, then we see that 
\begin{equation}
N^{\textup{R}}_{\CP}(\rho) \geq -\tr(W\rho) = 1/(9\sqrt{5}+9) \approx 0.0343. 
\end{equation} 

To further help us bound $N^{\textup{R}}_{\CP}$, we also introduce the robustnesses with respect to the sets $\DNN$ of doubly non-negative and $\DD$ of entrywise non-negative diagonally dominant density matrices:
\begin{equation}
    \DD \defeq \Big\{ \ \rho \in \DNN : \rho_{j,j} \geq \sum_{i\neq j} \rho_{i,j} \ \text{for all} \ j \ \Big\}.
\end{equation} 
These sets have the desirable property that we can optimize over them via semidefinite programming, so the following variants of $N^{\textup{R}}_{\CP}$ are efficiently computable:
\begin{equation}
    N^{\textup{R}}_{\DNN}(\rho) \defeq \min_{\sigma\in \D}\left\{s\geq 0\,:\, \frac{\rho+s\sigma}{1+s}\in \DNN\right\}  \label{eq:dnn_robustness}\\ 
\end{equation}
and    
\begin{equation}    
    N^{\textup{R}}_{\DD}(\rho) \defeq \min_{\sigma\in \D}\left\{s\geq 0\,:\, \frac{\rho+s\sigma}{1+s}\in \DD\right\}.\label{eq:dd_robustness}
\end{equation} 

Furthermore, these sets provide inner and outer approximations of the set of completely positive density matrices in the sense that 
\begin{equation} 
\DD \subseteq \CP \subseteq \DNN 
\end{equation} 
(with the first inclusion being the main result of \cite{Kay87}), so it immediately follows that 
\begin{equation} 
N^{\textup{R}}_{\DNN}(\rho) \leq N^{\textup{R}}_{\CP}(\rho) \leq N^{\textup{R}}_{\DD}(\rho) 
\end{equation} 
for all $\rho \in \D$. We furthermore have equality on the left when $n \leq 4$. It is perhaps worth making it clear, however, that  $N^{\textup{R}}_{\DNN}(\rho)$ and $N^{\textup{R}}_{\CP}(\rho)$ do not typically equal each other when $n \geq 5$, even if we restrict them to pure states, as demonstrated by the following example.

\begin{example}\label{exam:fourier_column_five}
    Let $\omega = e^{2i\pi/5}$ be the primitive fifth root of unity and consider the pure state 
    \begin{equation} 
        \ket{v} = (1,\omega,\omega^2,\omega^3,\omega^4)/\sqrt{5}. 
    \end{equation} 
    We claim that
    \begin{equation}
        N^{\textup{R}}_{\DNN}(\ketbra{v}{v}) = (3+\sqrt{5})/2 \approx 2.6180 < 2.8197 \approx 14 - 5\sqrt{5} = N^{\textup{R}}_{\CP}(\ketbra{v}{v}).
    \end{equation} 
    
    This value of $N^{\textup{R}}_{\DNN}(\ketbra{v}{v})$ can be found numerically via standard semidefinite programming software like CVX \cite{CVX} and can be proved analytically via standard semidefinite programming duality techniques (see \cite{Wat18}, for example).
    
    The fact that $N^{\textup{R}}_{\CP}(\ketbra{v}{v}) \leq 14 - 5\sqrt{5}$ follows from the fact that $N^{\textup{R}}_{\DD}(\ketbra{v}{v}) = 14 - 5\sqrt{5}$ (which can again be proved via standard semidefinite programming techniques). Finally, the fact that $N^{\textup{R}}_{\CP}(\ketbra{v}{v}) \geq 14 - 5\sqrt{5}$ can be seen by letting $W = I - 5(3 - \sqrt{5})\ketbra{v}{v}$ in the dual optimization problem~\eqref{eq:cp_robustness_dual}. In particular, it is then the case that $-\tr(W\ketbra{v}{v}) = 14 - 5\sqrt{5}$, and $W$ is a feasible point of that optimization problem because $W \preceq I$ trivially and the real part of $W$ is a non-negative scalar multiple of a Horn-like copositive matrix from Equation~\eqref{eq:horn_like_Wx}: $\Re(W) = (\sqrt{5}-2)W_{(3+\sqrt{5})/2}$. It follows that $W \in \CP^\circ$ by our discussion in Section~\ref{sec:copositive}. Note that $W$ is complex---no \emph{real} member $W$ of $\CP^\circ$ attains this same objective value of $-\tr(W\ketbra{v}{v}) = 14 - 5\sqrt{5}$. 
\end{example}

It is worth emphasizing that the above example is somewhat surprising and contrasts with the robustness of entanglement, where for pure states the robustness with respect to the set of separable states coincides with the robustness with respect to the set of PPT states \cite[Appendix~B]{VT99} (and both can be computed by a simple function of that pure state's Schmidt coefficients). It thus seems natural to ask whether or not $N^{\textup{R}}_{\CP}$ simplifies in any meaningful way when applied to pure states. The following theorem shows that it can be computed in terms of the $1$-norm of non-negativity.

\begin{theorem}\label{thm:rob_on_pure}
    For all pure states $\ket{v} \in \C^n$ we have $N^{\textup{R}}_{\CP}(\ketbra{v}{v}) = \big(\|\ket{v}\|_1^{\textup{N}}\big)^2 - 1$.
\end{theorem}

\begin{proof}
    This result follows immediately from \cite[Theorem~10]{Reg18}. In particular, we choose (in the notation of that paper) $\mathcal{V}$ to be the set of all pure states whose entries have equal phases as each other.
\end{proof}
  
In particular, the above result, together with convexity (property (C3)) of $N_{CP}^R$, immediately implies $N^{\textup{R}}_{\CP}(\rho) \leq \min\{n-1,3\}$ for all $\rho \in \D_n$, and this bound is tight in all dimensions for the exact same reason that the bound of Theorem~\ref{thm:nn_norm_bounds} is tight. The dimension-independence of this bound also contrasts with the robustnesses of coherence and entanglement, which can become arbitrarily large as the dimension $n$ increases.

The following theorem shows that the robustness of non-negativity simplifies even further, right down to an explicit formula, when applied to qubits. 

\begin{theorem}\label{thm:rob_qubits}
    If $\rho \in \D_2$ then $N^{\textup{R}}_{\DNN}(\rho) = N^{\textup{R}}_{\CP}(\rho) = N^{\textup{R}}_{\DD}(\rho) = 2\big|\min\{\Re(\rho_{1,2}),0\} + i\Im(\rho_{1,2})\big|$.
\end{theorem}

\begin{proof}
    Since we have $N^{\textup{R}}_{\DNN}(\rho) \leq N^{\textup{R}}_{\CP}(\rho) \leq N^{\textup{R}}_{\DD}(\rho)$ for all $\rho \in \D$, it suffices to show that \begin{equation} 
    N^{\textup{R}}_{\DD}(\rho) \leq 2\big|\min\{\Re(\rho_{1,2}),0\} + i\Im(\rho_{1,2})\big| \leq N^{\textup{R}}_{\DNN}(\rho). 
    \end{equation} 
    
    For the left inequality, we simply note that we can choose
    \begin{align}\label{eq:sigma_rob_qub_proof}
        s\sigma = \begin{bmatrix}
            \big|\min\{\Re(\rho_{1,2}),0\} + i\Im(\rho_{1,2})\big| & \min\{\Re(\rho_{1,2}),0\} + i\Im(\rho_{1,2}) \\
            \min\{\Re(\rho_{1,2}),0\} - i\Im(\rho_{1,2}) & \big|\min\{\Re(\rho_{1,2}),0\} + i\Im(\rho_{1,2})\big|
        \end{bmatrix}
    \end{align}
    so that the $(1,2)$ and $(2,1)$-entries of $\rho + s\sigma$ both equal $\max\big\{\Re(\rho_{1,2}),0\big\}$. Then ${(\rho + s\sigma)/(1+s) \in \DD_2}$ and $s = \tr(s\sigma) = 2\big|\min\{\Re(\rho_{1,2}),0\} + i\Im(\rho_{1,2})\big|$, so 
    \begin{equation}
    N^{\textup{R}}_{\DD}(\rho) \leq 2\big|\min\{\Re(\rho_{1,2}),0\} + i\Im(\rho_{1,2})\big|. 
    \end{equation} 
    
    For the right inequality, we just note that it is clear that the $(1,2)$-entry of the $s\sigma$ in Equation~\eqref{eq:sigma_rob_qub_proof} is as small as possible (in absolute value) subject to the constraint that $(\rho + s\sigma)/(1+s) \in \DNN_2$, and if we fix the $(1,2)$-entry of a positive semidefinite matrix then its trace is minimized when its diagonal entries are both equal to the absolute value of that $(1,2)$-entry. It follows that 
    \begin{equation}
    N^{\textup{R}}_{\DNN}(\rho) \geq 2\big|\min\{\Re(\rho_{1,2}),0\} + i\Im(\rho_{1,2})\big|, 
    \end{equation} 
    which completes the proof.
\end{proof}

The above theorem perhaps suggests defining an easy-to-compute measure of non-negativity $N^{\ell_1}$ via
\begin{equation}
    N^{\ell_1}(\rho) \defeq \sum_{i,j=1}^n \big|\min\{\Re(\rho_{i,j}),0\} + i\Im(\rho_{i,j})\big|,
\end{equation} 
in analogy with the $\ell_1$-norm of coherence \cite{BCP14}. While this measure equals the robustness of non-negativity when $n = 2$, it is not faithful (C1b) when $n \geq 5$, and it is not monotone (C2) even just under CPCP channels when $n \geq 3$. To see why, consider the channel $\Phi(X) = A_1XA_1^* + A_2XA_2^*$ and state $\rho \in \D_3$ given by 
\begin{equation} 
    A_1 = \begin{bmatrix}
        1/\sqrt{2} & 0 & 0 \\
        1/\sqrt{2} & 0 & 0 \\
        0 & 1 & 0
    \end{bmatrix}, \quad A_2 = \begin{bmatrix}
        0 & 0 & 0 \\
        0 & 0 & 0 \\
        0 & 0 & 1
    \end{bmatrix}, \quad \text{and} \quad \rho = \frac{1}{2}\begin{bmatrix}
        1 & -1 & 0 \\
        -1 & 1 & 0 \\
        0 & 0 & 0
    \end{bmatrix}.
\end{equation} 
It is straightforward to verify that $\Phi$ is a CPCP channel (after all, its Kraus operators are entrywise non-negative), but $N^{\ell_1}(\Phi(\rho)) = \sqrt{2} > 1 = N^{\ell_1}(\rho)$.


\subsection{The trace distance of non-negativity}\label{sec:trace_dist_non}

We now introduce a somewhat more geometrically-motivated measure of non-negativity, which asks how close the given density matrix is to the set $\CP_n$ of free density matrices. We define the \textit{trace distance of non-negativity} (in analogy with the trace distances of entanglement \cite{EAP03} and coherence \cite{RPL16}) by
\begin{align}\label{eq:trace_dist_nonneg}
    N_{\CP}^{\textup{tr}}(\rho) \defeq \min_{\sigma \in \CP} \big\{ \| \rho - \sigma \|_{\textup{tr}} \big\},
\end{align}
where $\| \rho - \sigma \|_{\textup{tr}}$ is the trace norm of $\rho - \sigma$ (i.e., the sum of the singular values of $\rho - \sigma$).

The fact that $N_{\CP}^{\textup{tr}}$ satisfies properties freeness (C1), faithfulness (C1b), and convexity (C3) are all straightforward to see. To see that it is montonic (C2), we just note that if $\Phi$ is a CP-preserving quantum channel and $\tilde{\sigma} \in \CP$ attains the minimum in Equation~\eqref{eq:trace_dist_nonneg} then
\begin{equation} 
        N_{\CP}^{\textup{tr}}\big(\Phi(\rho)\big) = \min_{\sigma \in \CP} \big\{ \| \Phi(\rho) - \sigma \|_{\textup{tr}} \big\} \leq \|\Phi(\rho) - \Phi(\tilde{\sigma})\|_{\textup{tr}} \leq \|\rho - \tilde{\sigma}\|_{\textup{tr}} = N_{\CP}^{\textup{tr}}(\rho),
\end{equation} 
with the second inequality coming from the fact that quantum channels cannot increase the trace norm.

The only remaining property of $N_{\CP}^{\textup{tr}}$ that remains to be determined is strong monotonicity (C2b). While we do not have an explicit counter-example to this property, it seems unlikely to hold, as the trace distances of coherence and entanglement are known to not be strongly monotonic \cite{YZXT16}. However, the following modification of $N_{\CP}^{\textup{tr}}$ where we instead consider the closest \emph{unnormalized} completely positive matrix to $\rho$ is indeed strongly monotonic (and also satisfies properties (C1b) and (C3) for the same reasons that $N_{\CP}^{\textup{tr}}$ does). Indeed, this was shown in \cite{Reg18}, where the upcoming quantity $N_{\lambda\CP}^{\textup{tr}}(\rho)$ that we introduce equals the quantity that they called $T^\prime_{S^{+}}(\rho)$, in the special case when $S = \mathcal{CP}$:
\begin{equation}
    N_{\lambda\CP}^{\textup{tr}}(\rho) \defeq \min_{\sigma \in \CP, \lambda \geq 0} \big\{ \| \rho - \lambda\sigma \|_{\textup{tr}} \big\}.
\end{equation} 

We now show that the trace distance of non-negativity and its modification both satisfy the same formula as the robustness of non-negativity when restricted to the $2$-dimensional case of qubits (refer back to Theorem~\ref{thm:rob_qubits}).   

\begin{theorem}\label{thm:trace_dist_qubits}
    If $\rho \in \D_2$ then $N_{\CP}^{\textup{tr}}(\rho) = N_{\lambda\CP}^{\textup{tr}}(\rho) = 2\big|\min\{\Re(\rho_{1,2}),0\} + i\Im(\rho_{1,2})\big|$.
\end{theorem}

\begin{proof}
    It is clear that $N_{\CP}^{\textup{tr}}(\rho) \geq N_{\lambda\CP}^{\textup{tr}}(\rho)$ in all dimensions, so it suffices to prove that 
    \begin{equation} 
    N_{\lambda\CP}^{\textup{tr}}(\rho) \geq 2\big|\min\{\Re(\rho_{1,2}),0\} + i\Im(\rho_{1,2})\big| \geq N_{\CP}^{\textup{tr}}(\rho). 
    \end{equation} 
    
    For the right inequality, we simply note that we can choose
    \begin{equation}
        \sigma = \begin{bmatrix}
            \rho_{1,1} & \max\big\{\Re(\rho_{1,2}),0\big\} \\
            \max\big\{\Re(\rho_{1,2}),0\big\} & \rho_{2,2}
        \end{bmatrix}
    \end{equation} 
    so that
    \begin{align}\label{eq:trace_norm_qubit_proof}
        \rho - \sigma = \begin{bmatrix}
            0 & \min\{\Re(\rho_{1,2}),0\} + i\Im(\rho_{1,2}) \\
            \min\{\Re(\rho_{1,2}),0\} - i\Im(\rho_{1,2}) & 0
        \end{bmatrix},
    \end{align}
    which has 
    \begin{equation} 
    N_{\CP}^{\textup{tr}}(\rho) \leq \|\rho - \sigma\|_{\textup{tr}} = 2\big|\min\{\Re(\rho_{1,2}),0\} + i\Im(\rho_{1,2})\big|. 
    \end{equation} 
    
    For the left inequality, we just note that it is clear that the $(1,2)$-entry of $\rho-\lambda\sigma$ from Equation~\eqref{eq:trace_norm_qubit_proof} (with $\lambda = 1$) is as small as possible (in absolute value) subject to the constraint that $\lambda\sigma \in \CP_2$, and if we fix the $(1,2)$-entry of a $2 \times 2$ matrix then its trace norm is minimized when its diagonal entries are both equal to each other and smaller in absolute value than that of the $(1,2)$-entry. It follows that 
    \begin{equation}
    N_{\lambda\CP}^{\textup{tr}}(\rho) \geq 2\big|\min\{\Re(\rho_{1,2}),0\} + i\Im(\rho_{1,2})\big|, 
    \end{equation} 
    which completes the proof.
\end{proof}

As with the robustness of non-negativity, we can get efficiently-computable upper and lower bounds on $N_{\CP}^{\textup{tr}}(\rho)$ and $N_{\lambda\CP}^{\textup{tr}}(\rho)$ by instead minimizing the trace distance to the sets $\mathcal{DDN}$ and $\DD$. Alternatively, the semidefinite programming hierarchy of \cite{Par00}, for example, can be used to construct semidefinite programs that compute any of these measures to as much accuracy as we like (though the size of those semidefinite programs grows quickly with the desired accuracy). 


\section{Relationship with coherence}\label{sec:coherence}

This resource theory is analogous to the resource theory of coherence \cite{BCP14} in many ways. In that resource theory, the free states are those that are ``incoherent'', which simply means that they are diagonal when represented in the computational basis. We denote this set of states by $\mathcal{I}$, and we note that it is trivially the case that $\mathcal{I} \subset \CP$. That is, every density matrix that is free in the resource theory of coherence is necessarily free in this resource theory of non-negativity as well. For this reason, we can think of the present resource as a sub-resource of coherence.

As a consequence of the inclusion $\mathcal{I} \subset \CP$, most of the properties of the resource theory of non-negativity are naturally bounded by an analogous property of the resource theory of coherence. For example, the measures of non-negativity that we introduced in the previous section all have analogous measures of coherence that are defined simply with the set $\CP$ replaced by $\mathcal{I}$. In particular, the robustness of coherence $C^{\textup{R}}$ \cite{NBCPJA16}, trace distance of coherence $C_{\textup{tr}}$ \cite{RPL16}, modified trace distance of coherence $C_{\textup{tr}}^\prime$ \cite{YZXT16}, and $\ell_1$-norm of coherence $C^{\ell_1}$ \cite{BCP14} satisfy the (trivial) bounds
\begin{align}
    N^{\textup{R}}_{\CP}(\rho) & \leq C^{\textup{R}}(\rho) & N^{\textup{tr}}_{\CP}(\rho) & \leq C_{\textup{tr}}(\rho) \\
    N^{\ell_1}(\rho) & \leq C^{\ell_1}(\rho) & N^{\textup{tr}}_{\lambda\CP}(\rho) & \leq C_{\textup{tr}}^\prime(\rho).
\end{align} 

These coherence measures have the advantage of being efficiently computable by semidefinite programming, as well as having numerous theoretic results known about them (see \cite{PCB16,CGJ16,JLP18,CF18} and the references therein, for example), so all of these results immediately provide bounds on the corresponding quantities concerning non-negativity.


\section{Conclusions and open questions} \label{sec:conclusions}

In this work, we introduce a resource theory for non-negativity of amplitudes of quantum states, motivated by the Sign Problem and stoquastic Hamiltonians. We showed that the free states in this resource theory are the well-studied completely positive matrices from linear algebra and convex optimization, and we characterized the accompanying witnesses and free operations.

We also introduced numerous ways of measuring how resourceful a quantum state is in this resource theory. Most of these measures are difficult to compute, so we also proved numerous bounds, and presented a method of approximating these measures via semidefinite programming.

Our work leaves numerous questions unanswered, and opens the door to many possible directions of future research, including:

\begin{itemize}
    \item The set of doubly non-negative density matrices that are not completely positive are mathematically directly analogous to the set of entangled density matrices with positive partial transpose (see \cite{Yu16,TAQLS17} for a way of making this relationship explicit). Since PPT states are bound entangled, it seems natural to guess that DNN-but-not-CP states are ``bound'' in some sense for this resource theory as well, and this seems worth exploring.
    
    \item Is there a nice operational interpretation of the set of CPCP channels? For instance, it is known that $J(\Phi)$ is separable if and only if $\Phi$ is a measure-and-prepare channel. Is there an analogous statement that can be made if $J(\Phi)$ is instead completely positive?
    
    \item We showed that the states $(I \pm Y)/2$ are maximally non-non-negative in the sense that they can be mapped via CP-preserving channels to arbitrary qubit states. Are there states that are similarly maximally non-non-negative in higher dimensions, and if so, what are they?
    
    \item We mentioned that the trace distance of coherence $N_{\CP}^{\textup{tr}}$ is monotonic, but probably not strongly monotonic. Can an explicit example be constructed to show that it indeed is not strongly monotonic?
    
    \item There are numerous other natural measures of non-negativity that could be defined and explored. For example, we could define the \emph{relative entropy of non-negativity} by
    \begin{equation}
        N_{\CP}^{\textup{r.e.}}(\rho) \defeq \min_{\sigma \in \CP} \big\{ S(\rho \| \sigma) : \mathrm{range}(\rho) \subseteq \mathrm{range}(\sigma)\big\},
    \end{equation} 
    where $S(\rho \| \sigma) = \tr\big(\rho\log(\rho)\big) - \tr\big(\rho\log(\sigma)\big)$ is the relative entropy of $\rho$ with respect to $\sigma$, and explore what properties and interpretations it has.\bigskip
\end{itemize} 
  
\noindent \textbf{Acknowledgements.} N.J.\ was supported by NSERC Discovery Grant number RGPIN-2016-04003. 

\bibliographystyle{ieeetr}
\bibliography{bib}

\begin{thebibliography}{10}

\bibitem{GW80}
L.~J. Gray and D.~G. Wilson, ``Nonnegative factorization of positive
  semidefinite nonnegative matrices,'' {\em Linear Algebra and Its
  Applications}, vol.~31, pp.~119--127, 1980.

\bibitem{DG14}
P.~J.~C. Dickinson and L.~Gijben, ``On the computational complexity of
  membership problems for the completely positive cone and its dual,'' {\em
  Computational Optimization and Applications}, vol.~57, pp.~403--415, 2014.

\bibitem{Ber88}
A.~Berman, ``Complete positivity,'' {\em Linear Algebra and Its Applications},
  vol.~107, pp.~57--63, 1988.

\bibitem{BS03}
A.~Berman and N.~Shaked-Monderer, {\em Completely Positive Matrices}.
\newblock World Scientific, 2003.

\bibitem{Yu16}
N.~Yu, ``Separability of a mixture of dicke states,'' {\em Physical Review A},
  vol.~94, p.~060101(R), 2016.

\bibitem{TAQLS17}
J.~Tura, A.~Aloy, R.~Quesada, M.~Lewenstein, and A.~Sanpera, ``Separability of
  mixed {D}icke states: an {NP}-hard optimization problem,'' {\em Quantum},
  vol.~2, p.~45, 2018.

\bibitem{MATS21}
C.~Marconi, A.~Aloy, J.~Tura, and A.~Sanpera, ``Entangled symmetric states and
  copositive matrices,'' {\em Quantum}, vol.~5, p.~561, 2021.

\bibitem{JM19}
N.~Johnston and O.~MacLean, ``Pairwise completely positive matrices and
  conjugate local diagonal unitary invariant quantum states,'' {\em Electronic
  Journal of Linear Algebra}, vol.~35, pp.~156--180, 2019.

\bibitem{SN21}
S.~Singh and I.~Nechita, ``Diagonal unitary and orthogonal symmetries in
  quantum theory,'' {\em Quantum}, vol.~5, p.~519, 2021.

\bibitem{PSVW18}
A.~Prakash, J.~Sikora, A.~Varvitsiotis, and Z.~Wei, ``Completely positive
  semidefinite rank,'' {\em Mathematical Programming}, vol.~171, pp.~397--431,
  2018.

\bibitem{Ohz17}
M.~Ohzeki, ``Quantum {M}onte {C}arlo simulation of a particular class of
  non-stoquastic hamiltonians in quantum annealing,'' {\em Scientific Reports},
  vol.~7, p.~41186, 2017.

\bibitem{CG18}
E.~Chitambar and G.~Gour, ``Quantum resource theories,'' {\em Reviews of Modern
  Physics}, vol.~91, p.~025001, 2019.

\bibitem{Vid00}
G.~Vidal, ``Entanglement monotones,'' {\em Journal of Modern Optics}, vol.~47,
  pp.~355--376, 2000.

\bibitem{BCP14}
T.~Baumgratz, M.~Cramer, and M.~B. Plenio, ``Quantifying coherence,'' {\em
  Physical Review Letters}, vol.~113, p.~140401, 2014.

\bibitem{VMGE14}
V.~Veitch, S.~A.~H. Mousavian, D.~Gottesman, and J.~Emerson, ``The resource
  theory of stabilizer quantum computation,'' {\em New Journal of Physics},
  vol.~16, p.~013009, 2014.

\bibitem{HG18}
A.~Hickey and G.~Gour, ``Quantifying the imaginarity of quantum mechanics,''
  {\em Journal of Physics A: Mathematical and Theoretical}, vol.~51, p.~414009,
  2018.

\bibitem{WKR21}
K.-D. Wu, T.~V. Kondra, S.~Rana, C.~M. Scandolo, G.-Y. Xiang, C.-F. Li, G.-C.
  Guo, and A.~Streltsov, ``Resource theory of imaginarity: Quantification and
  state conversion,'' {\em Physical Review A}, vol.~103, p.~032401, 2021.

\bibitem{Wat18}
J.~Watrous, {\em The Theory of Quantum Information}.
\newblock Cambridge University Press, 2018.

\bibitem{NC00}
M.~A. Nielsen and I.~L. Chuang, {\em Quantum computation and quantum
  information}.
\newblock Cambridge University Press, 2000.

\bibitem{Cho75}
M.-D. Choi, ``Completely positive linear maps on complex matrices,'' {\em
  Linear Algebra and Its Applications}, vol.~10, pp.~285--290, 1975.

\bibitem{Reg18}
B.~Regula, ``Convex geometry of quantum resource quantification,'' {\em Journal
  of Physics A: Mathematical and Theoretical}, vol.~51, no.~4, p.~045303, 2018.

\bibitem{HN63}
M.~Hall and M.~Newman, ``Copositive and completely positive quadratic forms,''
  {\em Proceedings of the Cambridge Philosophical Society}, vol.~59, p.~32933,
  1963.

\bibitem{BB03}
F.~Barioli and A.~Berman, ``The maximal {CP}-rank of rank k completely positive
  matrices,'' {\em Linear Algebra and Its Applications}, vol.~363, pp.~17--33,
  2003.

\bibitem{SBJS13}
N.~Shaked-Monderer, I.~M. Bomze, F.~Jarre, and W.~Schachinger, ``On the
  {CP}-rank and minimal {CP} factorizations of a completely positive matrix,''
  {\em SIAM Journal on Matrix Analysis and Applications}, vol.~34, no.~2,
  pp.~355--368, 2013.

\bibitem{TCF19}
G.~Torlai, J.~Carrasquilla, M.~T. Fishman, R.~G. Melko, and M.~P.~A. Fisher,
  ``Wavefunction positivization via automatic differentiation,'' {\em Physical
  Review Research}, vol.~2, p.~032060(R), 2020.

\bibitem{BV04}
S.~Boyd and L.~Vandenberghe, {\em Convex optimization}.
\newblock Cambridge University Press, 2004.

\bibitem{Par00}
P.~A. Parrilo, {\em Structured Semidefinite Programs and Semialgebraic Geometry
  Methods in Robustness and Optimization}.
\newblock PhD thesis, California Institute of Technology, 2000.

\bibitem{LR21}
J.~Levick and M.~Rahaman, ``Positively factorizable maps,'' {\em Linear Algebra
  and its Applications}, vol.~631, pp.~282--307, 2021.

\bibitem{Pau03}
V.~I. Paulsen, {\em Completely bounded maps and operator algebras}.
\newblock Cambridge University Press, 2003.

\bibitem{YMG16}
B.~Yadin, J.~Ma, D.~Girolami, M.~Gu, and V.~Vedral, ``Quantum processes which
  do not use coherence,'' {\em Physical Review X}, vol.~6, p.~041028, Nov 2016.

\bibitem{CG16}
E.~Chitambar and G.~Gour, ``Comparison of incoherent operations and measures of
  coherence,'' {\em Physical Review A}, vol.~94, p.~052336, 2016.

\bibitem{SuppCode}
N.~Johnston, ``{MATLAB} code for computing norms and measures of
  non-negativity.''
  \url{http://www.njohnston.ca/publications/res-theory-non-neg/}. Also
  available in the ``source'' files for the arXiv version of this paper, 2021.

\bibitem{VT99}
G.~Vidal and R.~Tarrach, ``Robustness of entanglement,'' {\em Physical Review
  A}, vol.~59, pp.~141--155, 1999.

\bibitem{NBCPJA16}
C.~Napoli, T.~R. Bromley, M.~Cianciaruso, M.~Piani, N.~Johnston, and G.~Adesso,
  ``Robustness of coherence: An operational and observable measure of quantum
  coherence,'' {\em Physical Review Letters}, vol.~116, p.~150502, 2016.

\bibitem{Kay87}
M.~Kaykobad, ``On nonnegative factorization matrices,'' {\em Linear Algebra and
  Its Applications}, vol.~96, pp.~27--33, 1987.

\bibitem{CVX}
M.~Grant and S.~Boyd, ``{CVX}: {MATLAB} software for disciplined convex
  programming, version 2.0 beta.'' http://cvxr.com/cvx, Sept. 2012.

\bibitem{EAP03}
J.~Eisert, K.~Audenaert, and M.~B. Plenio, ``Remarks on entanglement measures
  and non-local state distinguishability,'' {\em Journal of Physics A:
  Mathematical and General}, vol.~36, p.~5605, 2003.

\bibitem{RPL16}
S.~Rana, P.~Parashar, and M.~Lewenstein, ``Trace-distance measure of
  coherence,'' {\em Physical Review A}, vol.~93, p.~012110, 2016.

\bibitem{YZXT16}
X.-D. Yu, D.-J. Zhang, G.~F. Xu, and D.~M. Tong, ``Alternative framework for
  quantifying coherence,'' {\em Physical Review A}, vol.~94, p.~060302(R),
  2016.

\bibitem{PCB16}
M.~Piani, M.~Cianciaruso, T.~R. Bromley, C.~Napoli, N.~Johnston, and G.~Adesso,
  ``Robustness of asymmetry and coherence of quantum states,'' {\em Physical
  Review A}, vol.~93, p.~042107, 2016.

\bibitem{CGJ16}
J.~Chen, S.~Grogan, N.~Johnston, C.-K. Li, and S.~Plosker, ``Quantifying the
  coherence of pure quantum states,'' {\em Physical Review A}, vol.~94,
  p.~042313, 2016.

\bibitem{JLP18}
N.~Johnston, C.-K. Li, and S.~Plosker, ``The modified trace distance of
  coherence is constant on most pure states,'' {\em Journal of Physics A:
  Mathematical and Theoretical}, vol.~51, p.~414010, 2018.

\bibitem{CF18}
B.~Chen and S.-M. Fei, ``Notes on modified trace distance measure of
  coherence,'' {\em Quantum Information Processing}, vol.~17, p.~107, 2018.

\end{thebibliography}
\end{document}